\documentclass[letterpaper,10pt]{article} 
\usepackage{opticameet3} %% use version 3 for proper copyright statement

\newcommand\authormark[1]{\textsuperscript{#1}}

\usepackage{amsmath,amssymb}
\usepackage[colorlinks=true,bookmarks=false,citecolor=blue,urlcolor=blue]{hyperref} %pdflatex

\usepackage{graphicx}
\usepackage{subcaption}
\usepackage{amsmath}
\usepackage{amsthm}
\usepackage{comment}
\usepackage{placeins}
\usepackage{booktabs} 
\usepackage{colortbl} 
\usepackage{xcolor}
\usepackage[square,numbers]{natbib}

\newtheorem{prop}{Proposition}
\newtheorem{defi}{Definition}

\newtheorem{cor}{Corollary}

\begin{document}

\title{Imaging-based representation and stratification of intra-tumor Heterogeneity via tree-edit distance}

% \author{Author name(s)}
% \address{Author affiliation and full address}
% \email{e-mail address}
%%Uncomment the following line to override copyright year from the default current year.
%\copyrightyear{2022}

\author{Lara Cavinato,\authormark{1,*} Matteo Pegoraro,\authormark{1} Alessandra Ragni,\authormark{1}, Martina Sollini,\authormark{2,3}, Paola Anna Erba,\authormark{4,5} and Francesca Ieva\authormark{1,6}}

\address{\authormark{1} MOX, Politecnico di Milano, Deparment of Mathematics, Milan, 20133, Italy \\
\authormark{2} Department of Biomedical Sciences, Humanitas University, Pieve Emanuele, Italy \\
\authormark{3} IRCCS Humanitas Research Hospital, Rozzano, Italy. \\
\authormark{4} Regional Center of Nuclear Medicine, Department of Translational Research and Advanced Technologies in Medicine and Surgery, University of Pisa, 56126, Pisa, Italy \\
\authormark{5} Medical Imaging Center, University of Groningen, University Medical Center Groningen, Groningen, The Netherlands \\
\authormark{6} Human Technopole, Health Data Science Center, 20157, Milan, Italy}

\email{\authormark{*}lara.cavinato@polimi.it} %% email address is required

\begin{abstract}
Personalized medicine is the future of medical practice. In oncology, tumor heterogeneity assessment represents a pivotal step for effective treatment planning and prognosis prediction. Despite new procedures for DNA sequencing and analysis, non-invasive methods for tumor characterization are needed to impact on daily routine. On purpose, imaging texture analysis is rapidly scaling, holding the promise to surrogate histopathological assessment of tumor lesions. In this work, we propose a tree-based representation strategy for describing intra-tumor heterogeneity of patients affected by metastatic cancer. We leverage radiomics information extracted from PET/CT imaging and we provide an exhaustive and easily readable summary of the disease spreading. We exploit this novel patient representation to perform cancer subtyping according to hierarchical clustering technique. To this purpose, a new heterogeneity-based distance between trees is defined and applied to a case study of Prostate Cancer (PCa). Clusters interpretation is explored in terms of concordance with severity status, tumor burden and biological characteristics.
Results are promising, as the proposed method outperforms current literature approaches.
Ultimately, the proposed methods draws a general analysis framework that would allow to extract knowledge from daily acquired imaging data of patients and provide insights for effective treatment planning.
\end{abstract}

\section{Introduction}
\label{sec:introduction}

The current paradigm shifting of modern medical practice sinks its root in providing personalized treatments and improving therapy outcomes. Huge strides have been made in oncology with the uprising of quantitative imaging techniques and new procedures for DNA sequencing and analysis that allow an extensive characterization of cancer subtypes.
In particular, recent research has investigated the main causes of cancer progression, resistance to therapy and late recurrence.
Among these, tumor heterogeneity has gained special interest and has been recognized to play a crucial role \cite{fisher2013cancer}:
defined as complex genetic, epigenetic and protein modifications that can be found within the same patient's disease, tumor heterogeneity behaves as a driver for phenotypic selection. According to Stanta and Bonin and y Cajal et al. \cite{stanta2018overview, y2020clinical}, different types of tumor manifestation may exist as a response to microenvironmental and external changing, differing between primary tumor and proximal and distant metastases. As a result, certain tumor phenotypes properly respond to therapies and others become resistant clones, leading to treatments ineffectiveness and cancer progression. 
Pertinently, detecting at baseline which phenotype will respond and which will not - known as \textit{prognostic cancer subtyping} - represents a pivotal step in personalized medicine.

Although recent findings about heterogeneity suggest that therapy would be improved if guided by the analysis of both primary and metastatic tissues  - such as lymph nodes \cite{cummings2014metastatic} -, clinical practice usually relies on primary tumor biomarkers for prognosis definition and treatment planning. Thus, baseline assessment emerges altered by the understimation of intra-tumor heterogeneity which behaves as confounding factor in pre-treatment clinical-pathological prognosis, leading to poor survival rates \cite{esparza2017breast}.
This misalignment between research evidence and clinical practice seems mostly due to the lack of non-invasive methods for heterogeneity quantification. Accordingly, current prognostic cancer subtyping cannot be translated into daily clinical practice and therapeutic guidelines.
%  Accordingly, this is the rationale behind this work: suggesting a direction which could help filling such gap.

Over the last two decades, the texture analysis of digital images - such as Magnetic Resonance Imaging (MRI) and Positron Emission Tomography / Computer Tomography (PET/CT) - has arisen as a valuable non-invasive proxy for biological assessment of tumors, eventually growing in a discipline of its own, namely radiomics \cite{mayerhoefer2020introduction}
Specifically, macroscopic appearance of tumors has been acknowledged as a valid tool for guiding clinical decisions in the definition of disease severity and treatment planning.
Broadly speaking, image texture analysis consists of extracting descriptors of spatial variation of voxel grey-scale  and intensity within the image Volumes Of Interest (VOI), i.e., the tumor lesions.
Under the name of radiomic features, such textural descriptors form a high dimensional vector embedding of the VOI and may provide a non-invasive assessment of tumor appearance from routinely acquired imaging studies \cite{gillies2016radiomics}. These features are indeed supposed to supply additional predictive and prognostic information, ready to use to postulate the underlying biological mechanisms of disease progression in clinical routine \cite{chicklore2013quantifying}. 
Accordingly, the dissimilarity in the appearance of different lesions, therefore in their texture descriptions, can be regarded as \textit{radiological} heterogeneity, which can be easily quantified and leveraged in the daily practice.

Despite the increasing interest in tumor heterogeneity, imaging-guided therapy currently employs  biomarkers for tumor burden that stem from the characterization of the primary tumor, the bigger lesion (often coinciding with the hottest lesion) or the mean lesions' profile. Only recently few radiomics-based approaches have been suggested - for prognosis, treatment outcome and survival prediction - which consider the multi-lesion disease in a comprehensive way. In particular, several researchers \cite{eertink202218f,ceriani2020sakk38,burggraaff2020optimizing} proposed different segmentation strategies for feature extraction from patient level VOIs, while Cottereau et al. \cite{cottereau202018f} evaluated the predictive power of several indicators reflecting the spatial distributions of malignant \textit{foci} spread throughout the whole body. A number of \textit{dissemination} features have been explored and reviewed:  the number of lesions, the euclidean distance between crucial or predominant bulks, the largest value of the pairwise sum of the physical distances between lesions, etc.
Stemming from a similar idea, Cavinato et al. \cite{cavinato2020pet} proposed a similarity metric for comparing lesions’ texture descriptions, defining intra-patient heterogeneity as the normalized average of pairwise distances between lesions' radiomic vectors.
This similarity over patient's lesions description has thus been suggested as functional, rather than spatial, dispersion index for tumor burden and disease severeness, with promising results in Hodgkin Lymphoma \cite{sollini2020methodological} and Prostate Cancer \cite{sollini202118f}.
Preliminary results represent an insightful starting point in the debate around the proper definition of heterogeneous disease.

In this work, motivated by the need to embed tumor heterogeneity quantification into patients’ clinical pathway planning, we propose a novel way for modeling intra-patient tumor heterogeneity in a non-invasive way, leveraging the radiomic framework. 
Specifically, we perform dimensionality reduction on radiomic vectors, as to remove redundancy and collinearity while preserving the multi-view nature of the texture description. Reduced vectors of peer lesions within the same tumor are then compared via pairwise distances.
Representing the patient via the pairwise distance matrix of its lesions makes it laborious to compare patients with different numbers of lesions. For this reason, upon lesions’ distance matrix, we build a dendrogram, which hierarchically aggregates peer lesions in a unique combinatorial object. This object-oriented representation summarizes the multi-lesion disease and highlights the evolutionary relationship among lesions, basing on similarities in their imaging characteristics. In fact, lesions are not independent as they are statistically and semantically connected to the patient they belong to. Accordingly, such relationship shapes and influences the structure of the dendrogram associated to the patient.
We then exploit the tree-based patient representation to cluster cancer subtypes according to their imaging heterogeneity. To do so, we define a new \textit{ad hoc} distance between trees. % detailing formulations and proofs. 
To validate the method, we test the whole pipeline on a dataset of patients affected by metastatic Prostate Cancer (PCa), evaluating the descriptive and stratification performance in terms of disease severeness and outcomes. We associate imaging subtypes to clinically relevant information within and beyond clinical surrogates, with the goal of eventually supporting therapy decisions wherein actions regarding active surveillance, mild treatment or intensified therapy are devised and taken \cite{fisher2013cancer}. 

\section{Results}

\subsection{Case study: Prostate Cancer}

Within the personalized medicine framework, Prostate cancer (PCa) is a striking example of the need to exploit an insightful prognostic cancer subtyping for treatment planning. In fact, even if recent studies have reported a decreasing pattern of overall PCa incidence, Culp et al. \cite{siegel2020prostate} and Siegel et al. \cite{culp2020recent} recorded an alarming mortality rate due to an increasing trend of distant stage metastatic disease, even in developed countries.
Moreover, the role of imaging-guided therapy for PCa has revealed to be very promising and is consistently spreading in daily practice \cite{giovacchini201011c}. Despite these facts, clinical guidelines still relies on primary tumor biomarkers. Besides, very limited methods have been proposed for reliably assessing and quantifying multi-lesion heterogeneity information within the same patient from an imaging point of view.
This misalignment between research evidence and clinical routine results in poor disease free survival rates, mostly due to the lack of non-invasive methods for heterogeneity quantification.

The case study analyzed in this work is composed by a set of $N=333$ lesions belonging to fifty-five patients of Azienda Ospedaliero-Universitaria Pisana with multi-site, multi-lesion, recurrent Prostate Cancer confirmed with a positive %[$^{18}$F]FMCH
PET/CT study.
% All patients were enrolled  between January 2011 and February 2018 at Azienda Ospedaliero-Universitaria Pisana.
The study was performed in accordance with the Declaration of Helsinki and approved by the local ethics committee. The signature of a specific informed consent and the legal requirements of clinical trials were waived given the observational retrospective study design.
During the observational trial, patients showed evidence of biochemical recurrence after first-line treatments, exhibiting metastatic disease.
Every patient manifested a different number of tumor lesions $n_i$, according to the spreading burden of the metastatic tumor. Information about age, sex, lesion site, total tumor volume, Gleason Score \cite{epstein20162014}, Prostate Specific Antigen \cite{balk2003biology} and therapy treatment was collected per each patient. 
Personal information and qualitative tumor data are displayed in Table \ref{tab:summary_continuous} and Table \ref{tab:summary_categorical}.
Additionally, from PET/CT, volumes of interest, i.e. lesions, were segmented by experienced nuclear medicine physicians and texture features were extracted over VOIs according to the radiomic framework, resulting in forty radiomic features ($p=40$). Both the segmentation of lesions and radiomic features extraction were performed using LifeX software \cite{nioche2018lifex}, according to the formulas detailed in the software documentation (www.lifexsoft.org).

\begin{figure}[t]
    \includegraphics[width=\textwidth]{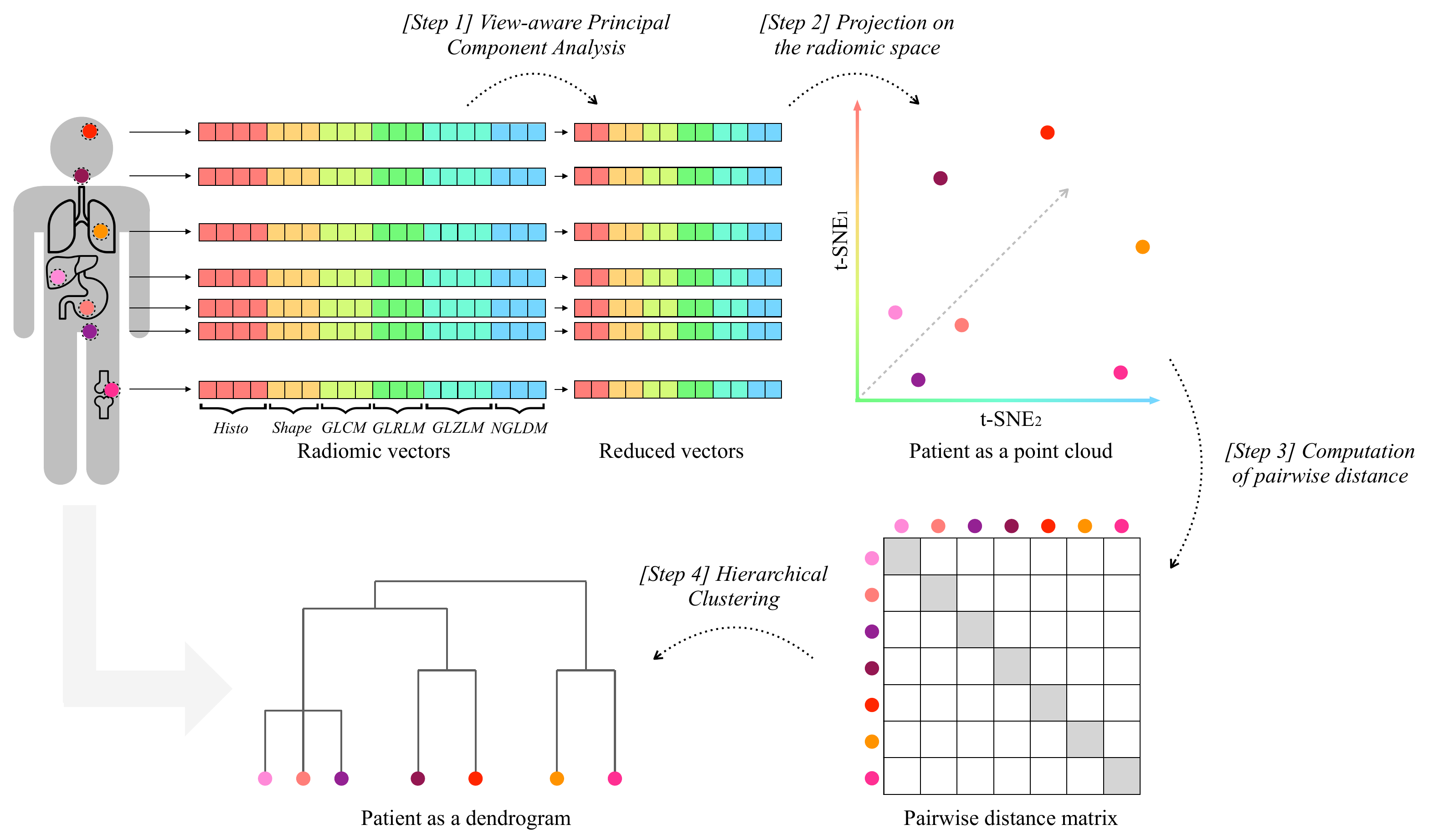}
    \caption{Patient representation pipeline: lesions’ radiomic vectors of each patient are dimensionally reduced according to view-aware Principal Component Analysis. 
    \textit{[Step 1]} Features are grouped according to the six semantic group, or \textit{view}, they are semantically divided into.
    As to preserve a balanced importance between views, two principal components are kept from the scores of each PCA, leading to different percentages of explained variability. 
    A total of twelve principal components results from the process, which include six orthogonal pairs of linear combinations of original features.
    \textit{[Step 2]} Accordingly, patients are represented as finite sets of $n_i$ points in $\mathbb{R}^{12}$, that is the reduced radiomic space according to view-aware strategy implementation. In the example, $n_i=7$.
    \textit{[Step 3]} Pairwise (Euclidean) distance is compute among patients’ lesions and \textit{[Step 4]} hierarchical clustering with \textit{average} linkage is applied to distance matrices, resulting in a dendrogram $T$ representing each patient.}
    \label{fig:pipeline} 
\end{figure}

We fed Prostate Cancer imaging data into the pipeline described in Fig. \ref{fig:pipeline}, obtaining a tree-based representation $T$ for each of the patients.
The pruned edit distance $d_P^\mu$, as defined in the Methods, was implemented and leveraged to compute the patient-to-patient distance matrix. % on which to perform stratification for the case study of interest. 
Clustering of patients was thus completed according to hierarchical clustering algorithm with the proposed \textit{ad hoc} distance and \textit{ward} linkage. Specifically, \textit{ward} linkage was chosen as it provided better results under a prognostic point of view with respect to other linkages (for a definition of all available linkages see Appendix \ref{sec:build_dendro}).
The number of clusters was selected over the range $[2,5]$, as a trade off between similarity performance and interpretability. Specifically, the number $k$ that presented both a reasonable silhouette coefficient and a high concordance (in terms of mutual information) with therapy response was selected.
The resulting classes could then be intended as groups of patients with similar representations in terms of heterogeneous disease, to be characterized according to exogenous clinical variables and risk assessment.
% Clinical relevance and implications of the methodology follow this Section and are elucidated in the discussion.

\subsection{Clusters characterization}
\label{sec:comparison}

As to profile the clustering, we describe how the stratification procedure captures the differentiation of tumor heterogeneities and provide a clinical/biological interpretation.

Upon pipeline implementation, hierarchical clustering identified three groups: groups 0, 1 and 2 hosted 39, 10 and 6 patients respectively.
In Fig. \ref{fig:gPCA_heights} the curves of the heights of the trees' vertices over the three groups can be appreciated: branches present different average heights according to the group their dendrograms belong (see Fig. \ref{fig:gPCA_heights}). Groups are shown to entail different heterogeneity extent, following an ANOVA functional approach \cite{pini2017fun_test} \cite{horvath2012inference}.

Beside the group-wise characterization of tree conformation as manifestation of tumor heterogeneity, clinical variables were used as exogenous factors to characterize and interpret the groups. We used appropriate tests according to the variable type, normality of data and sample size. Normality was tested according to the Shapiro test. We thus employed Mann-Whitney non-parametric tests for comparing distributions of continuous (non-normal) variables; parametric t-tests for testing the difference of means in continuous (normal) variables; Levene non-parametric tests for comparing variances of continuous (non-normal) variable; Bartlett parametric tests for continuous (normal) variable ratio of variances; $Chi-squared$ tests for independence of categorical variable. 
P-values are indicated respectively as $p_{m/d}$ for tests on means/distributions, $p_{var}$ for tests on variance and $p_{ind}$ for tests on independence.
% In some cases, non-parametric tests could not be performed due to the presence of ties. In those cases, parametric tests were performed instead, so that particular attention should be payed when interpreting p-values, which could be underestimated.
Pairwise one-sided comparison between groups rather than multivariate analysis was investigated as to provide a group-wise characterization. As to avoid potential Type II errors due to small sample size, value of $\alpha=0.1$ was considered for significance.

We evaluated the differences between the obtained groups in terms of number of oligo/multi-metastatic patients (as classified with two different clinical cut-offs of 3 and 5 lesions), number of patients with bone disease, total tumor volume and number of tumor lesions. Also, the implementation of combined therapy (such as joint radiotherapy and chemotherapy with respect to only chemotherapy) and response to therapy were evaluated in patients of different groups.
Additionally, among clinical prognostic tools, tumor aggressiveness is usually assessed with Gleason Grading System (or Gleason Score) \cite{epstein2016contemporary}. A Gleason Score (GS) is given to Prostate Cancer based upon its microscopic appearance with respect to cell differentiation. Pathological scores represent the sum of the primary and secondary patterns (each ranging from 1 - well differentiated, like normal cells - and 5 - poorly differentiated, i.e., abnormal cells) and range from 2 to 10. Higher numbers indicate more aggressive disease, worse prognosis and higher mortality \cite{epstein20162014}. In particular, patients with Gleason Score exceeding the value of 7 experience extraprostatic extension and biochemical recurrence more frequently than others \cite{draisma2006gleason}.
Accordingly, clusters were also analyzed in terms of mean Gleason Score and number of patients exceeding GS of 7.

Besides, Prostate Specific Antigen (PSA) has been proposed for screening, assessment of future risk of prostate cancer development, detection of recurrent disease after local therapy and treatment planning of advanced disease. Often employed as criteria in combination of stage and GS, its role in early stage assessments is still debated due to instability of measurements and the presence of confounding factors. However, PSA is still considered a valid tool for prognosis and treatments in advanced stages of metastatic prostate cancer \cite{pezaro2014prostate}. 
Moreover, PSA values after cytotoxic regimens has been shown to predict survival. Particularly, the decrease in PSA levels is associated to therapy response in soft tissue lesions and thus could be intended as a proxy of therapy outcome \cite{smith1998change}. 
Accordingly, we recorded PSA levels before the therapy (PSA0), right after the first line of therapy (PSA1) and at the end of the follow up (PSA2). Delta-PSA levels were computed between PSA1-PSA0 and PSA2-PSA0 as proxies of cancer evolution. In the following, they will be referred as PSA, $\Delta PSA_{1,0}$ and $\Delta PSA_{2,0}$.

% As for exogenous clinical variables, 
Table \ref{tab:results_gPCA} and Fig. \ref{fig:gPCA_curves} elucidate the results.
The profile of the blue and green groups are very similar for what PSA ($p_{m/d}=0.3787$, $p_{var}=0.4714$) and $\Delta PSA_{1,0}$ ($p_{m/d}=0.3477$, $p_{var}=0.4533$) are concerned, with a very limited range of values concentrated around zero.
Different trends are exhibited by the blue and green curves of the $\Delta PSA_{2,0}$ ($p_{m/d}=0.0591$), where the difference could support the hypothesis of different cancer evolution starting from similar baseline assessments. Yet, they present similar variance ($p_{var}=0.2159$).
% The orange group, on the other hand, presents a wider range of PSA0, delta PSA1-PSA0 and delta PSA2-PSA0 values than the other two groups, being more heterogeneous.
The orange group, on the other hand, presents wider ranges and higher intra-group heterogeneity.
In particular, orange PSA is significantly higher than the blue group with a much more spread distribution ($p_{m/d}=0.0116$; $p_{var}=0.0013$) yet no statistical difference with the green groups is confirmed ($p_{m/d}=0.3089$; $p_{var}=0.1845$); orange $\Delta PSA_{1,0}$ is significantly lower than the blue group ($p_{m/d}=0.0019$) but not than the green one ($p_{m/d} = 0.1810$), however its distribution appears more spread and inhomogeneous, covering both the negative and the positive axis, in both cases ($p_{var}=0.0003$; $p_{var}= 0.0995$). The $\Delta PSA_{2,0}$ of the orange group does not vary from the one of the blue group ($p_{m/d}=0.3689$). However, it shows a higher variance than the other, suggesting a heterogeneous long-term tumor prognosis ($p_{var}=0.0066$).
Also, the orange group and the green group do not differ significantly in their average ($p_{m/d}=0.1855$) but their variances reveal a mild divergence in terms of distribution kurtosis ($p_{var}=0.1085$).

\begin{figure}[h!]
	\centering
	\begin{subfigure}[c]{\textwidth}
    	\centering
    	\includegraphics[width = \textwidth]{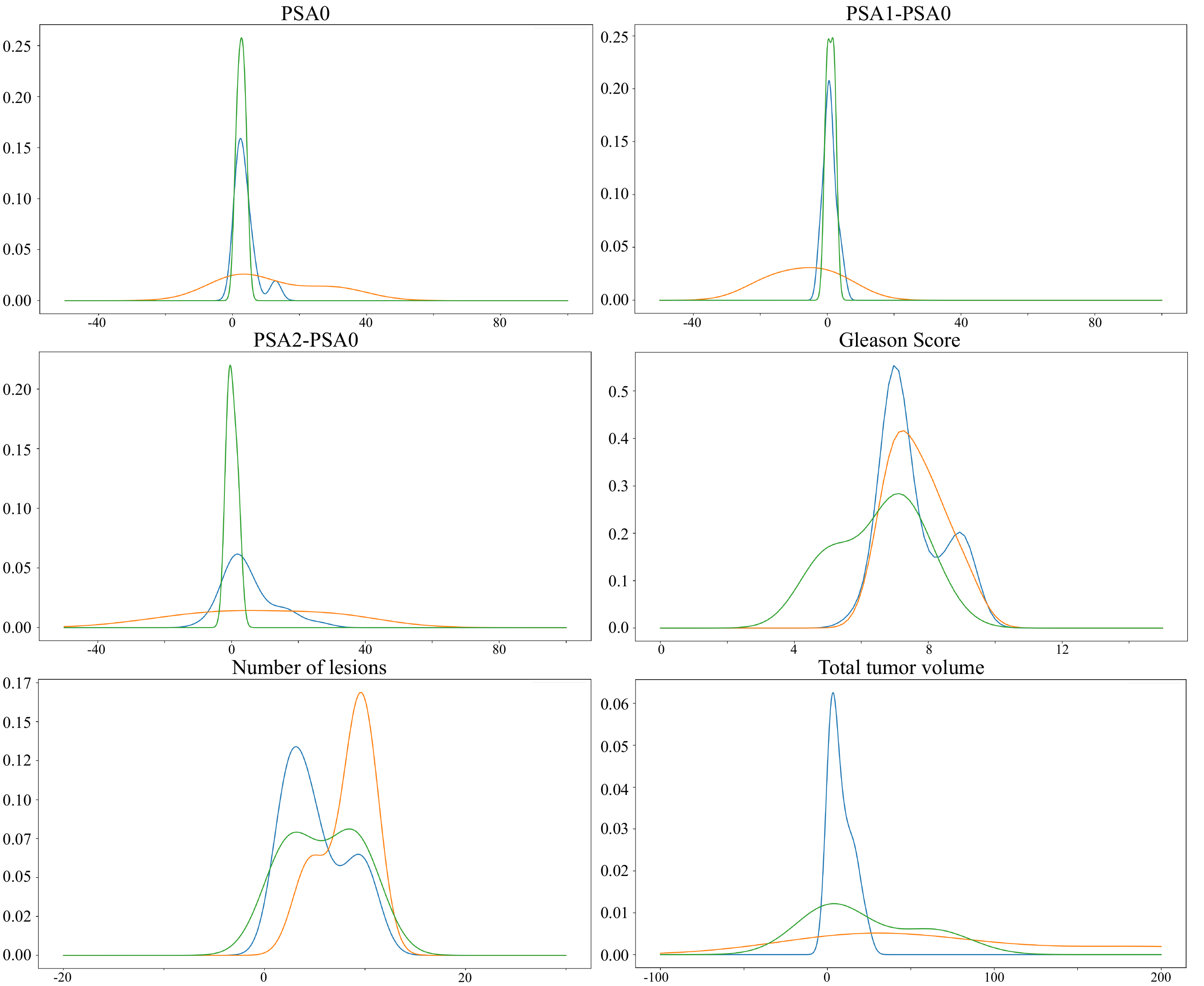}
    \end{subfigure}
	\begin{subfigure}[c]{\textwidth}
		\centering
		\includegraphics[width = \textwidth]{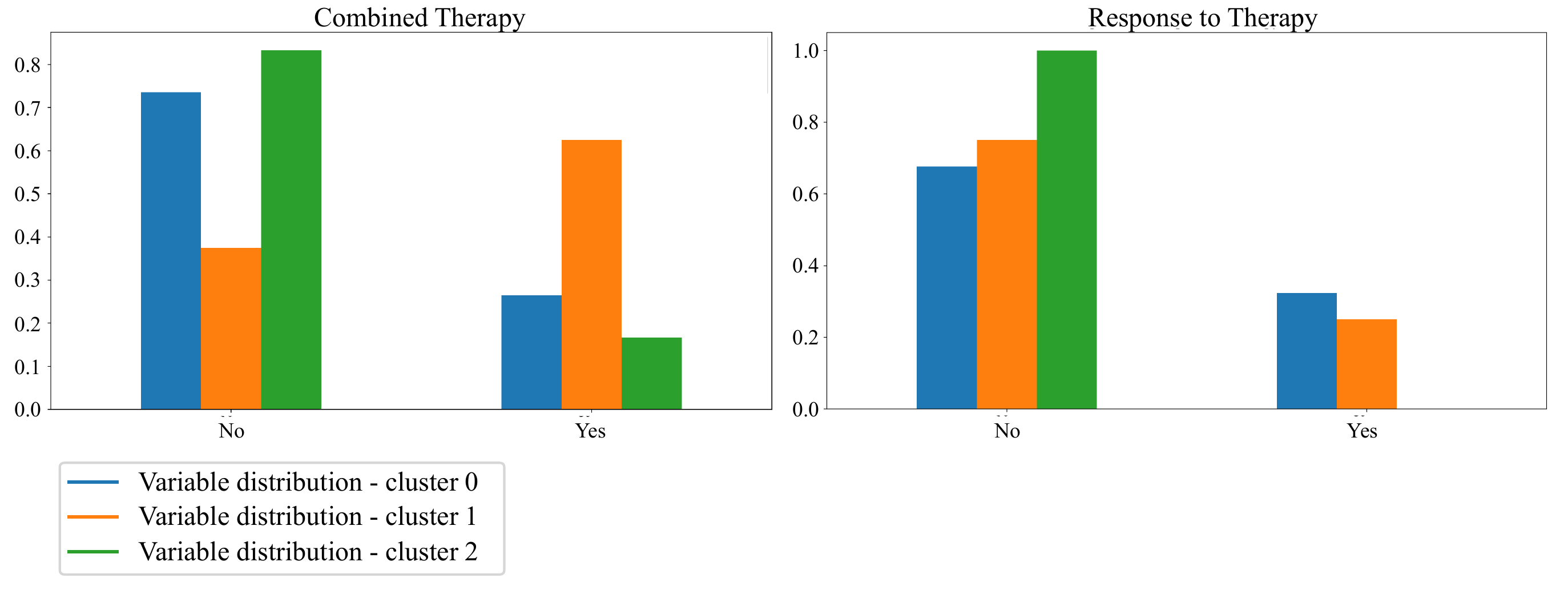}
	\end{subfigure}
    \caption{Results of clustering characterization: first three rows draw the distributions of the numerical clinical variables in the three groups, namely the PSA values, the $\Delta PSA_{1,0}$, the $\Delta PSA_{2,0}$, the number of lesions, Gleason Scores and the total tumor volume; last row shows the proportions of the categorical clinical variables in the three groups, that are the combination of therapy and the response to treatment. For the proportion of skeleton disease and of the oligo/multi-metastatic status as devised by the two clinical cut-offs (3 and 5 lesions) see Appendix \ref{sec:additional}.}
    \label{fig:gPCA_curves}
\end{figure}

Regarding the number of lesions, the orange group displays a higher number of metastases than the blue one ($p_{m/d}=0.0081$). The green group exhibits a behavior very similar to the blue group ($p_{m/d}=0.4162$), diverging from the orange group with respect to which it presents fewer lesions ($p_{m/d}=0.0722$).
Moreover, total volume of the tumor is related to the number of lesions. In fact, the blue group displays a reduced spreading of the tumor over the body with respect to the orange group ($p_{m/d}=0.0002$) but not to the green group ($p_{m/d}=0.4917$). The orange and the green groups also exhibit a statistical difference in terms of tumor volume ($p_{m/d}=0.0306$).
Of note, despite the number of metastases in the blue and green groups are very similar, it should be noticed that their tumor spreading appears shifted in the figure, entailing unrelated tumor burden information. Similarly, the orange group, while presenting a greater number of lesions, shows an extension of the tumor visually analogous to the green group. Such discrepancy is imputable to the difference of variances the distributions display. % (0 vs 1: $p_{var}=0.0000$; 0 vs 2: $p_{var}=0.0047$; 1 vs 2: $p_{m/d}=0.2009$).

From these consideration, it appears clear how the green group shows phenotypic similarities and dissimilarities with respect to both blue group and orange group, presenting an in-between behavior.
However, the detach of green patients from the rest of the population is mostly driven by the different distribution of GS levels.
In fact, the blue and orange groups do not show peculiar differences ($p_{m/d}=0.2967$), although both differ from the green group, compared to which they have a higher GS ($p_{m/d}=0.0419$; $p_{m/d}=0.0601$). As it will be further discussed in discussion, prognostic power of GS values should be taken with the grain of salt due to their qualitative and aggregated nature.

\begin{figure}[t]
	\centering
	\includegraphics[width = \textwidth]{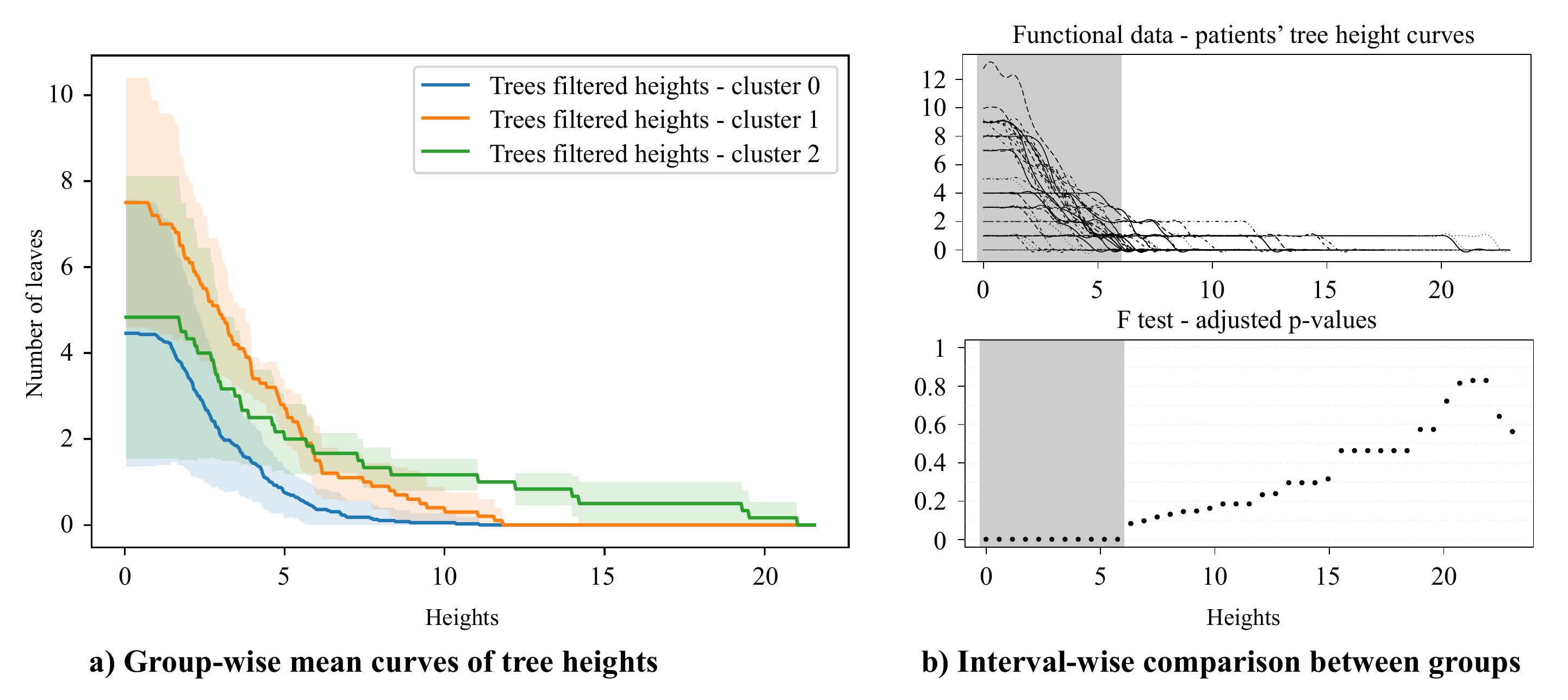}
    \caption{a) Curves displaying the \textit{filtered} heights of the trees' vertices for the three groups. Operationally, curves were built as follows: for any fixed height (x-axis), for any tree in the selected group, we count the number of nodes whose height value is greater than the fixed one (y-axis). For the step-by-step procedure see Appendix \ref{sec:building_curves}. The curves in the plot represent the pointwise within-group means of such counts, and the shaded regions cover an area of $1$ standard deviation around the means. The values of such counting process result in a monotonically non-increasing function detecting information about trees' heterogeneity. In fact, 
    higher values of such function, especially as the height threshold becomes bigger and bigger, correspond to a greater number of heterogeneous lesions in the patients.
    Patients of group 0 (blue line) are characterized by a very homogeneous disease where trees branches are on average less and very short compared to the other groups; patients of group 1 (orange line) tend to exhibit more lesions than patients belonging to group 0, lesions which are intermediately heterogeneous, as their representation trees display both short branches and longer branches than group 0; patients in group 2 (green line) are associated to very heterogeneous diseases, displaying a similar number of lesions to group 0, but with the associated branches being much longer. A synthetic example of tree per each group is displayed in Fig. \ref{fig:mu}, elucidating the differences with a graphical support.
    b) Functional comparison between curves: in order to test the hypothesis that curves belonging to different groups are different, we use the ANOVA procedure proposed in \cite{pini2017fun_test}. It outputs an interval-wise adjusted p-value function. Depending on the sort and level $\alpha$ of Type-I error control, significant intervals can be selected. Here, we highlighted in grey the region of significance. Of note, the curves appear different for what homogeneity-heterogeneity balance is concerned; they loose significance as they approach very big height values.
    }
\label{fig:gPCA_heights}
\end{figure}

As for the clinical assessment of patients, the blue and green groups present similar to each other yet opposite characterizations with respect to the orange group. 
They display a lower percentage of patient with bone disease ($p_{ind}=0.0769$; $p_{ind}=0.1729$), therefore fewer people who have undergone an invasive combination of therapies ($p_{ind}=0.0517$; $p_{ind}=0.0863$). 
Moreover, although the results on the response to therapy are not significant due to the limited data available, they reveal a certain trend. In fact, both blue and green groups of patients are administered a milder therapy with respect to orange group. On one hand, such treatment results effective for the blue group, which shows the highest percentage of responders; while, on the other hand, this is not the case for the green group, which manifests the highest percentage of non-responders. Group 2 thus exhibit a clinical characterization comparable to group 0, whereas tree conformation analysis and prognostic assessment, i.e., response to therapy, agree in granting it a higher score of risk. 
Finally, the orange group presents the highest number of multi-metastatic patients, %(0 vs 1: $p_{ind}=0.0601$),
followed by the blue group and finally the green group, which hosts mostly oligo-metastatic patients.

From Fig. \ref{fig:km_gPCA}, some extent of stratification is appreciable, although the groups’ survival curves separation is not neat and statistically significant ($p = 0.12$). All patients of group 0 gradually respond since they feature mild disease, both from a structural, i.e., tree conformation, and clinical point of view. The green group host patients who the clinic would treat as not severe (in terms of number of lesions, GS and PSA baseline information), but our radiomics investigation has put in an at risk group, to be properly monitored, in terms of tree structure and tumor extension. In line with the results of our policy, these patients do not respond to therapy during the study period. Finally, the orange group carries severe patients from both a structural and a clinical point of view. 

\begin{figure}[t]
	\centering
	\includegraphics[width = \textwidth]{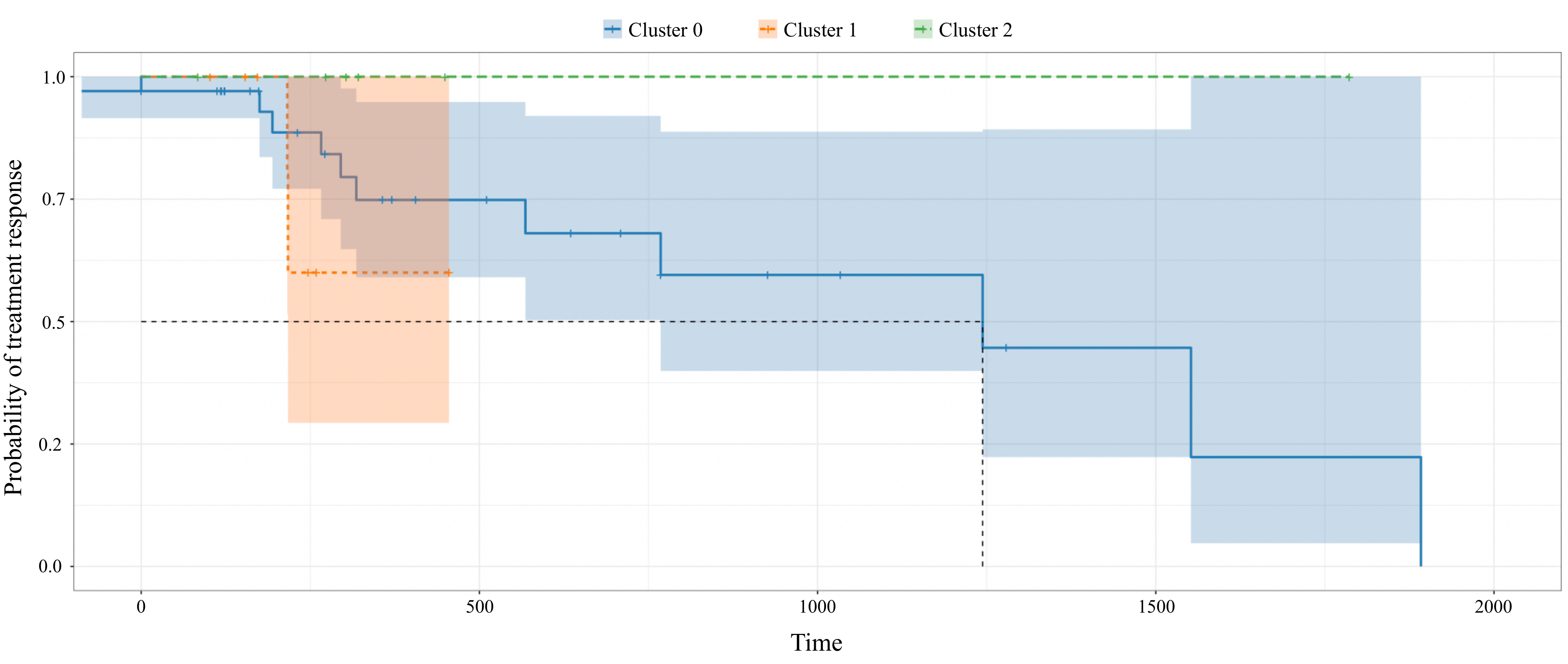}
    \caption{Group-wise Kaplan Meier curves of time to therapy response: it visually shows the probability of the response to treatment in a certain time interval. The blue line, the orange line and the green line correspond to group 0, 1 and 2 arising from clustering performed on patients' dendrograms. Groups have a different time to response. In particular, green group does not respond to therapy along the study period. Orange group shows indeterminate results due to the lack of and heterogeneity of clinical data. Blue group gradually responds throughout the study period.}
    \label{fig:km_gPCA}
\end{figure}

Since unsupervised approaches are thoroughly dataset dependent, hierarchical clustering grouped in the same clusters very heterogeneous patients, due to the limited data available. In fact, clinical variable variance of orange patients was consistently larger than other groups - despite not being the largest cluster. Interestingly, we fit a DBSCAN (Density Based Spatial Clustering of Applications with Noise) algorithm \cite{khan2014dbscan} on the pruned-edit distance matrix which lead to the same clustering policy of patients. In this setting, while blue and green groups were confirmed to be clusters with similar density, the orange group was classified as noise, i.e., observations that display inconsistent density characterization. 
Accordingly, a couple of patients responded to therapy while the majority did not respond and entered more invasive treatments. For these reasons, the orange survival curve is hardly interpretable and is left out the discussion. For sure, the high variability of this group testifies that a larger testing cohort would allow to identify further separations within this group, leading to clearer prognostic results. 

\begin{table}
\centering
    \begin{tabular}{llccc}
    \toprule 
    Variable & Test on & 0 vs 1 & 0 vs 2 & 1 vs 2 \\
     &  & (p-values) & (p-values) & (p-values) \\
    \midrule
    \rowcolor{black!10} GS & Mean & 0.2967 & \textbf{0.0419} & \textbf{0.0601} \\
    \rowcolor{black!10} & Variance & 0.8368 & 0.5433 & 0.7093 \\
    Gleason Category & Independence & 0.5129 & 0.5056 & 0.3077 \\
    \rowcolor{black!10} Oligo or Multi ($>3$) & Independence & \textbf{0.0601} & 0.9260 & 0.1729 \\
    Oligo or Multi ($>5$) & Independence & \textbf{0.0848} & 0.6868 & 0.3339 \\
    \rowcolor{black!10} $3<$Lesions$\leq5$ & Independence & 0.1969 & 0.9022 & 0.3950 \\
    N lesions & Mean & \textbf{0.0081} & 0.4162 & \textbf{0.0722} \\
     & Variance & 0.3871 & 0.4357 & 0.1469 \\
    \rowcolor{black!10} Skeleton & Independence & \textbf{0.0769} & 0.9622 & 0.1729 \\
    Total Volume (ml) & Mean & \textbf{0.0002} & 0.4917 & \textbf{0.0306} \\
     & Variance & \textbf{0.0000} & \textbf{0.0047} & 0.2009 \\
    \rowcolor{black!10} PSA & Mean & \textbf{0.0116} & 0.3787 & 0.3089 \\
    \rowcolor{black!10}  & Variance & \textbf{0.0013} & 0.4714 & 0.1845 \\
    $\Delta PSA_{1,0}$ & Mean & \textbf{0.0019} & 0.3477 & 0.1810 \\
     & Variance & \textbf{0.0003} & 0.4533 & \textbf{0.0995} \\
    \rowcolor{black!10} $\Delta PSA_{2,0}$ & Mean & 0.3689 & \textbf{0.0591} & 0.1855 \\
    \rowcolor{black!10} & Variance & \textbf{0.0066} & 0.2159 & \textbf{0.1085} \\
    Ongoing Therapy & Independence & \textbf{0.0601} & 0.5875 & 0.3339 \\
    \rowcolor{black!10} Combined Therapy & Independence & \textbf{0.0517} & 0.6091 & \textbf{0.0863} \\
    Therapy Response & Independence & 0.6856 & 0.127 & 0.2907 \\
    \bottomrule
    \end{tabular}
    \caption{Significance in terms of p-values of the statistical tests between cluster 0 and cluster 1, cluster 0 and cluster 2, cluster 1 and cluster 2 in the proposed pipeline: non-parametric/parametric tests on difference of averages and variances were performed for (non-normal/normal) numerical variables while tests on category independence were performed for categorical variables.}
    \label{tab:results_gPCA}
\end{table}

\subsubsection{Comparison with State-of-the-Art methods}
\label{sec:soa-comparison}

The established radiomics frameworks contemplate the extraction of texture features from a single lesion, often located on the prostate where the bigger lesion or the primary tumor are found. Such features are usually fed into a classification or stratification model as to predict cancer diagnosis, staging and prognosis.

As a comparison with the state of the art, we investigated the stratification resulting from the analysis of the biggest lesions' textural description. 
We selected the bigger lesion of each patient, we reduced the texture vector dimensionality according to view-aware PCA dimensionality reduction procedures and we performed hierarchical clustering on the patient-to-patient Euclidean distance matrix with \textit{ward} linkage.
The clustering procedure lead to the stratification of patients into two groups, namely group 0 and group 1. 
% As we also tested the one-lesion-pipeline with traditional PCA dimensionality reduction, the two groups resulting from PCA and view-aware PCA pipelines showed an agreement, i.e., Rand Index, of 70\%.
It is worth noting that this clustering approach - based only on the bigger lesion and/or primary tumor - share some extent of the stratification underpinnings of the tree-based clustering. For the sake of clarity, we refer to one-lesion clustering as \textit{tumor clustering} and to tree-based clustering as \textit{heterogeneity clustering}.
In particular, tumor clustering resulted to have a mild concordance with heterogeneity clustering (Rand Index $=0.43$ \cite{chacon2021close}). Coherently, the tumor-based stratification leads to clinical significance. 
Tumor clustering pipeline discriminated between patients with different GS ($p_{m/d}=0.0259$), number of lesions ($p_{m/d} = 0.0001$), oligo/multi-metastatic disease proportions ($p_{ind}=0.0191$), PSA ($p_{m/d}=0.0339$), ongoing therapy ($p_{ind}=0.0847$) and total volume ($p_{m/d} < 0.0001$). However, $\Delta PSA_{1,0}$ ($p_{m/d}=0.2942$), $\Delta PSA_{2,0}$ ($p_{m/d}=0.2920$), proportion of patients exhibiting bone disease ($p_{m/d}=0.5220$), combination of therapy ($p_{ind}=0.3698$) and response to therapy ($p_{ind}=0.2170$) did not result significant in tumor clustering pipeline.
These findings were somehow expected. In fact, therapeutic guidelines are mainly taken on the basis of the characterization of the primary tumor. Accordingly, these results confirm the role of the primary tumor in acting as a driver for tumor heterogeneity and enforce radiomics role in the clinical treatment planning.
Nevertheless, despite the coherence with qualitative clinical investigation, tumor-based stratification does not translate into a risk assessment and prediction. In fact, the Kaplan Meier curve, describing the probability of response to treatment of the two groups, appear almost superimposed ($p = 0.85$) and do not reveal any prognostic mechanism of the clustering.

As a step forward from one-lesion strategy, radiomics literature suggests to average radiomic descriptions of peer lesions belonging to a patient, as to obtain one single vector. 
Such vector-based representation plays for the mean imaging phenotype of all lesions expressed by a patient, taking into account the variability of the imaging profiles. Such method provide an information-complexity trade-off between one-lesion strategy and the tree-based patient representation we propose.
Under these considerations, we performed patient-wise weighting of lesions' vectors, implemented the view-aware PCA dimensionality reduction methods and computed vector-based representation of each patient.
The pipeline grouped all the patients in one cluster, although one patient with higher PSA was clustered separately from the rest of the cohort population as to meet hyperparameter criteria (e.g. minimum number of clusters at least equal to 2). 
Clear stratification was indeed not achieved in this setting, however a particularly bad-prognosis patient detached from the main group.
From these findings, it follows that vector-based representation model did not lead to clear and solid results in our dataset, suggesting the non robustness of the lesions’ weighting procedures. 
% Accordingly, the proposed tree-based representation outperforms consolidated methods.

\section{Discussion}
\label{sec:discussion}

Current radiomic framework presents some limitations, including the inter-operator variability in imaging acquisition settings, the relatively small sample sizes bounding the performance of supervised approaches, the lack of standardization, the high dimensionality and the collinearity of radiomics variables as well as the absence of a clinical interpretation for features \cite{smith2019radiomics}.
For these reasons, intra-patient tumor heterogeneity quantification has long been attempted with poorer results, hampering its embedding into daily practice. 
In this work, we propose a patient representation for agnostic multi-lesion cancer description, able to overcome intrisinc limitations of radiomics. The method exploits the texture analysis of lesions' imaging according to the radiomic workflow, overcoming features redundancy with PCA-based dimensionality reduction strategies. The proposed dendrogram representation results \textit{agnostic} with respect to acquisition settings and operator variability as it is built upon evolutionary and statistical relationship within peer lesions' descriptions. Moreover, the small sample size issue is tackled by the employment of unsupervised methods.
As to leverage the complex representation for stratification purposes, a suitable distance between dendrograms was required. 
Indeed, the pruned tree edit distance was specifically designed for heterogeneity-based hierarchical dendrograms and was the keystone to deliver a stratification policy based on agnostic disease conformations.

For what dimensionality reduction is concerned, view-aware PCA was hereby proposed as a scalable yet minor novelty. PCA is a well known, established, interpretable technique, often producing good results.  With the rationale of improving the scalability of the approach, we introduced the view-aware PCA, such that the dimensionality reduction step could be performed in parallel in smaller and semantically-similar euclidean spaces.
Nevertheless, further studies could investigate the sensitivity of the representation model and the performance of the clustering policy when employing feature transformation methods that capture non-linear dependencies of across- and within-view features \cite{li2018review}.

Compared to state-of-the-art disease representation, our approach shapes an exhaustive representation of intra-patient heterogeneity and devises an informed patient stratification. In fact, it leads to a more complex yet low-processed modelling of cancer disease, underlining interactions and relationships between lesions of individuals from which to infer prognostic knowledge.
Clearly, one-lesion strategy did not provide a quantification of lesions’ diverse phenotypes within a patient, as it only relies on the primary tumor. Nevertheless, tumor clustering lead to a coherent stratification with respect to the current clinical biomarkers, i.e., PSA, GS and oligo/multi-metastatic status. However, such clinically-informed stratification did not reach a significance in terms of prognostic power, bringing out the limitation of current clinical and radiomic-based biomarkers for treatment and prognosis.
Interestingly, the proposed representation brings out a comprehensive way to capture tumor biology and heterogeneity, revealing a deeper appreciation of the disease than a single lesion or the primary tumor alone.
On the other hand, the vector-based representation was confirmed insufficient to properly embed the patient's complexity of information. In fact, mean radiomic profile seems not to properly capture intra-tumor variability while it overlooks the primary tumor information entailing clinical information.
In both cases - when only the primary tumor is considered and when the mean radiomic profile of lesions is computed - state of the art methods failed in perspectively stratifying patients.

Beside descriptive and prognostic purposes, the proposed tree-based representation and stratification of tumor heterogeneity permits an exhaustive comparison between the role played by the primary lesion and its involvement into phenotypic selection mechanism. This is worth to be drawn and further investigated from a tumor heterogeneity and prognostic point of view. 
In fact, tumor clustering showed a latent agreement with heterogeneity clustering, suggesting the reliability of the current clinical practice in assessing intra-tumor characterization from primary lesions. Accordingly, primary tumor information seems to be more informative than intra-patient mean lesions’ profiles. If used in combination with dissemination indexes - such as number of metastases, dispersion of intra-patient lesions’ radiomic profiles and number of involved organs -, primary tumor characterization could provide enough information to support therapeutic decisions when an exhaustive assessment of tumor metastases results too expensive.

On note, heterogeneity clustering highlighted milder significance for what GS biomarkers is concerned with respect to tumor clustering. Pertinently, although GS is a solid clinical prognostic factor driving therapy planning, it represents the histo-pathological analysis for characterizing primary and secondary tumor biology at molecular level. Accordingly, the aggregated value, that is the sum of primary differentiation pattern and secondary differentiation pattern, do not entail heterogeneity information. For instance, studies using surrogate PCa end points have suggested that outcomes for GS 7 cancers vary according to the predominance of pattern 4. PCa mortality, biochemical progression and development of metastases differ for 3 + 4 and 4 + 3 tumors \cite{stark2009gleason}. This means that, according to tree-based representation, patients tagged with a GS 7 may still be clustered in different prognostic groups and alter the tests on averages. For these reasons, GS should not be considered as a solid ground truth for a perspective model, rather it conveys only a association between radiomic-based heterogeneity assessment and its biological counterpart, that is tumor microscopic appearance.
On the other hand, PSA and $\Delta PSA$ values significantly supported the predictive power of imaging-based representation in terms of cancer progression and disease free survival. Consistently, a decrease in PSA levels after treatment regimens was associated to therapy response.
In this sense, exhaustive lesions’ texture assessment and imaging-based heterogeneity quantification devise cancer subtypes that correlates with prognosis beyond clinical surrogates, eventually supporting treatment planning.

Basing on our and literature findings, the systematic digital tissue collection and its analysis should be enforced in the translational research of tumor disease and in the developing of targeted therapies. The debate around the therapeutic exploitation of imaging biomarkers for intra-tumor heterogeneity is nowadays on the cutting edge of medicine literature and it interlaces with other science field such as mathematics and geometry. This dynamic interplay between disciplines may provide a propitious route to ultimately attempt to limit tumor progression and treatment resistance.
Stemming from this work, future research could consider longitudinal evolution of heterogeneity-based representation objects and, accordingly, investigate the course of the disease over time in a non invasive way.

\section{Methods}
\label{sec:methods}

In this section we outline the steps involved in the proposed methodological pipeline. In particular, methods for radiomics-based representation of patients’ heterogeneity and its stratification are discussed. 
%Section \ref{sec:dim_reduction} presents
We present the challenges of analyzing a general radiomic dataset proposing an insightful dimensionality reduction approach (M1). Representation strategy is then deduced and described (M2). 
We then introduce an existing edit distance for comparing tree objects, on which we build the proposed metrics. It follows the derivation of an \textit{ad hoc} metric (M3) for capturing intra-tumor heterogeneity variability and computing the similarity matrix between patients on which to perform the stratification according to hierarchical clustering.

\subsection{M1: Dimensionality reduction}
\label{sec:dim_reduction}

As previously introduced, radiomic features are regarded as a high dimensional vector embedding of the VOI, providing a non-invasive assessment of tumor appearance from routinely acquired imaging studies. Several softwares, e.g. LifeX software, allow to extract several texture indexes from VOIs, according to the formulas provided by the software documentation (www.lifexsoft.org).
Considerable efforts have been devoted to link biological meaning with texture descriptors.
So far, little evidence of tight correlation between the two has been found, preventing from univocally define tumor inherent heterogeneity of lesions. However, different textural features have been proposed and reviewed by Castellano et al. \cite{castellano2004texture} as measures of tumor-specific intra-lesion heterogeneity. Indeed, radiomics analysis is widely assumed to entail all the information needed for a definition of lesion heterogeneity \cite{eary2008spatial, chaddad2018predicting}.

When managing a radiomic dataset, several challenges come across, above all high dimensionality and collinearity between features. Thus, prior to pairwise distance computation, lesions’ radiomic vectors need to be properly reduced as to selectively bring out relevant information. 

According to Nioche et al.\cite{nioche2018lifex}, radiomic features divide into six semantic groups of different methodological levels of texture analysis.
\textit{First order statistics} are the statistical moments of the grey level distribution extracted from the VOI under analysis. \textit{Shape features} describe morphological characteristics of the tumor.
The \textit{Grey Level Co-occurrence matrix} (GLCM) describes the co-occurrence of pairs of grey values in the VOI at a given distance $\delta$ (offset), usually set to 1, towards thirteen different directions. 
The \textit{Grey Level Run Length matrix} (GLRLM) describes the length of homogeneous \textit{runs} for each grey level, averaged across thirteen directions. 
Similarly, the \textit{Grey Level Zone Length matrix} (GLZLM) provides information on the size of homogeneous \textit{zones} for each grey level, averaged across three dimensions. 
Finally, the \textit{Neighbour Grey Level Difference matrix} (NGLDM) corresponds to the difference of grey levels between one voxel and its twenty-six neighbors in three dimensions. From each of these groups, several indices are extracted, exhibiting a multi-view intrinsic structure that induces intra- and inter-group correlation patterns. Accordingly, such vectors disclose high collinearity between their elements that needs to be properly managed. To overcome this, we leverage the very basic idea of multi-view learning and dimensionality reduction approaches: the view-wise linear combination of features \cite{kettenring1971canonical}.
We propose to separately apply the PCA to each of the radiomic groups, as to exploit the multi-view nature of the radiomic vectors. In this way, we may keep the information carried by each group well discerned, as it is methodologically extracted in different ways. A more interpretable dimensionality reduction comes from the process.

Upon pre-processing, namely missing values imputation and Z-transform normalization of radiomic variables, we thus perform this novel dimensionality reduction, namely \textit{view-aware PCA}.

As depicted in Figure \ref{fig:pipeline}, features are grouped according to the six semantic group - \textit{view} - as described above.
Within each group, PCA is performed and two principal components are retained from the scores of each PCA, resulting in different percentages of explained variability.
The process yields a total of twelve principal components, including six orthogonal pairs of linear combinations of original features. It follows that each lesion is described by a twelve-dimensional vector entailing view-wise texture information.

Further, we build the patient representation upon the such reduced radiomic vectors of peer lesions.

\subsection{M2: Tree-based patient representation}
\label{sec:pat_representation}

To exhaustively represent patients' disease in terms of tumor heterogeneity, relationships between lesions needs to be learnt from data.
Distance between texture descriptors could be an appropriate surrogate. Specifically, radiomic variables of a lesion - possibly after dimensionality reduction as in M1 - define a lesion-specific point in an Euclidean space. All lesions belonging to the same patient form a point cloud in $\mathbb{R}^p$, with a number of points $n_i$ equal to the number of patient’s tumor lesions and $p$ being the number of radiomic variables.

Although some frameworks are available to compare point clouds via discrete transport \cite{memoli2008gromov} \cite{nguyen2021point}, interpretability is often limited by the high dimensionality of the embedding space. %, constraining the explainability to leverage statistics computation that add processing layers and alter the complexity of the analysis. 
Also, model based approaches, which capture the variability of cloud-generating processes by means of interpretable parameters, require a high number of observations in each point cloud to produce reliable estimations \citep{ghosal@2017nonparam}.

A more insightful approach would be to transform the point cloud into a proper summary, i.e., a representation, equally informative and easily readable. Pertinently, hierarchical clustering dendrograms have been extensively studied in the last decades as they unveil the intrinsic relationship among points of a point cloud (for a review on hierarchical clustering dendrograms see \cite{dendro_1}).
In our setting, the rationale behind hierarchical clustering stems from the need to quantify to which extent lesions, i.e., their radiomic vectors, are similar within patients and how they get agglomerated, hierarchically, one to each other. 
A dendrogram is obtained in such a way that lesions are linked in terms of evolutionary relationship, based on similarities in their imaging characteristics. Fig. \ref{fig:tree} graphically describes the process while Appendix \ref{sec:build_dendro} formalizes the mathematical steps involved.
Dendrograms' structure reflects the homogeneity between points of the point cloud. For instance, Fig. \ref{fig:mu} presents three dendrograms: the blue one describes a condensed point cloud, the green one presents a scattered point cloud while the orange tree denotes a hybrid situation.

To build hierarchical clustering dendrograms, a similarity measure is needed together with an agglomerative criterion - also known as \emph{linkage} - that best suit the structure of the data and the aim of the analysis. In our setting, an appropriate similarity measure is the Euclidean distance between lesions’ radiomic vectors, as suggested by Cavinato et al. \cite{cavinato2020pet}. Additionally, \textit{average} linkage is employed as it is known to be less sensitive to outliers, producing a more robust representation \cite{jarman2020hierarchical}.

\begin{figure}[t]
    \includegraphics[width=\textwidth]{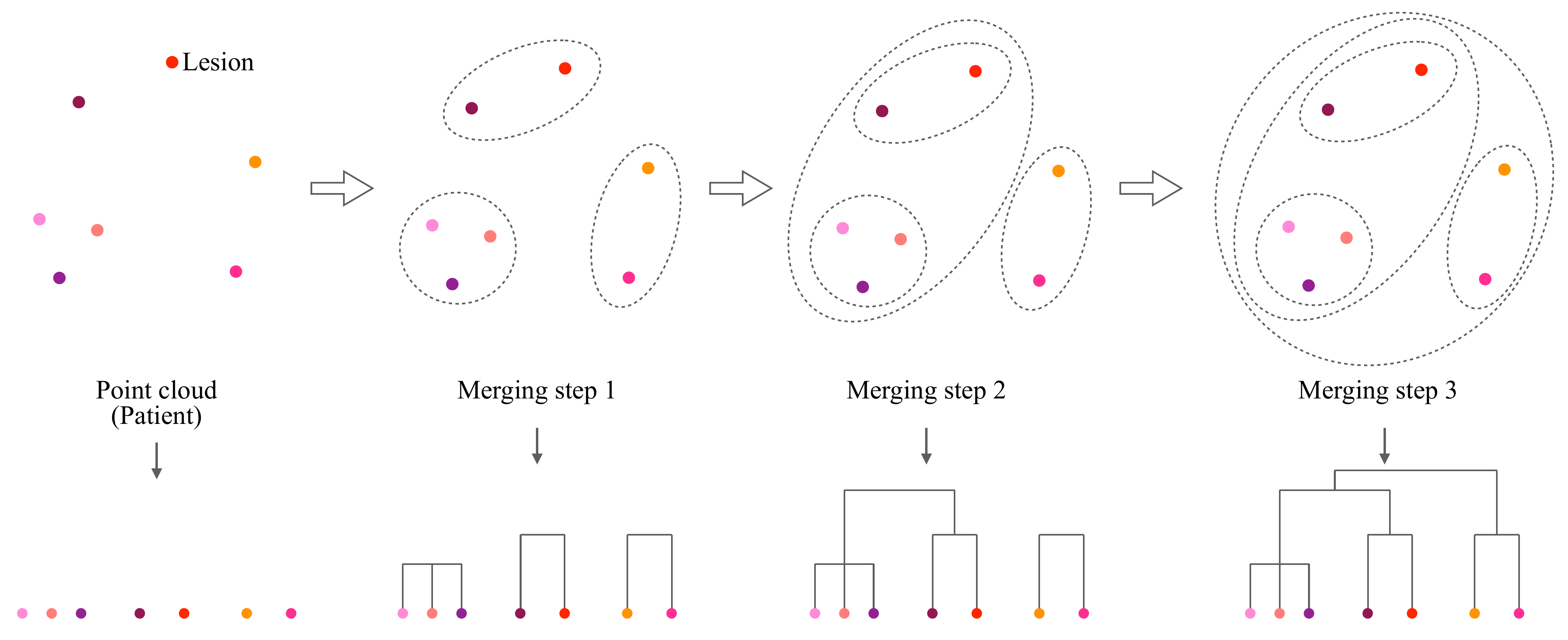}
    \caption{Tree-based patient representation via agglomerative hierarchical clustering: from the bottom up to the root, leaves get agglomerated and merged into bigger and bigger clusters, to finally converge in a single set. As a consequence, tree branches reflect pairwise similarity between lesions and the tree structure surrogates the overall dispersion among peer lesions. In the final dendrogram representation, leaves are the lesions of the patient and edges illustrate the similarity-connection between them. Leaves that are close to each other are intended by construction to be similar and exhibit a comparable radiomic profile (homogeneous) while distant leaves can be thought as lesions expressing different imaging phenotypes (heterogeneous). In this sense, dendrogram structure entails the heterogeneity quantification within the tumor, which needs to be exploited for heterogeneity-based stratification of patients. For mathematical formulation see Appendix \ref{sec:build_dendro}.}
    \label{fig:tree} 
\end{figure}

\subsection{M3: A novel Heterogeneity-based distance}

\begin{figure}[h!]
    \includegraphics[width=\textwidth]{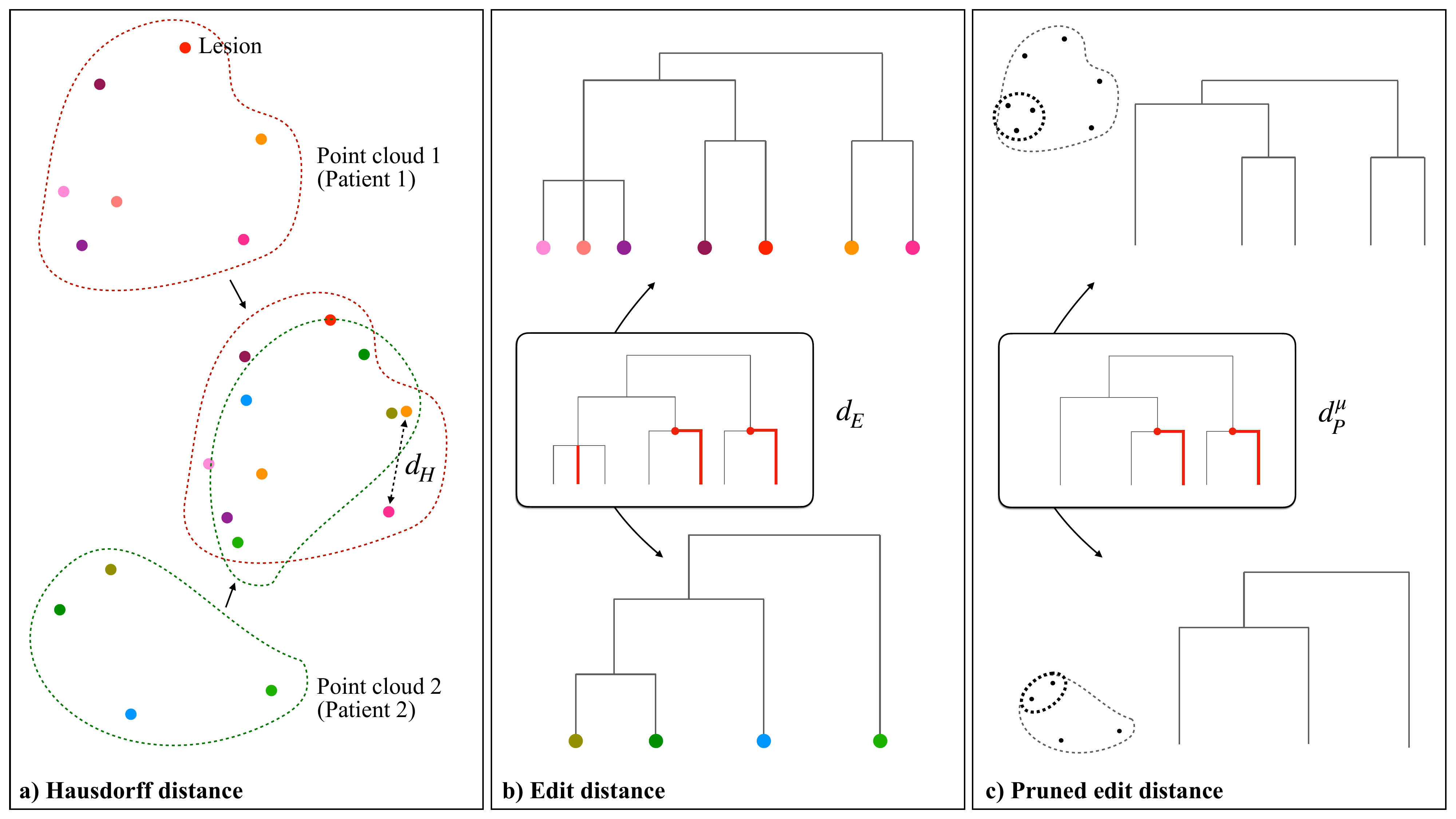}
    \caption{Continuity among metrics. a) Housdorff distance between two point clouds: the point clouds get overlapped and $d_H$ is defined as the maximum distance between the two maximally distant points; Hausdorff-closeness reflects the similarity in the spreading of points of two point clouds throughout the space. Specifically in the radiomic space, such spreading entails the quantification of inter-patient heterogeneity. This means that Hausdorff-close point clouds, i.e., patients' sets of lesions, have similar intra-patient heterogeneity characterization and thus should be regarded as similar by the metric we employ for dendrograms; b) Tree edit distance between hierarchical clustering dendrograms: the distance is given by the sum of the costs of the minimum number of modifications needed for transforming a tree into the other. Modifications include positive/negative shrinking, deletion/insertion and ghosting/splitting. 
    The {\em shrinking} edit multiplies the weight value of an edge with a positive factor, which can either lengthen (positive shrinking) of shorten (negative shrinking) the original edge weight. The cost of shrinking an edge is equal to the absolute value of the difference between the initial and the final weights.
    {\em Deleting} or {\em inserting} an edge $(v_1,v_2)$ removes or introduces a branch at a given height, altering the children-father structure of the tree. For any deletion/insertion, the cost is equal to the weight of the edge deleted/inserted.
    Finally, the {\em ghosting} edit eliminates a vertex \(v\) that connects only two adjacent edges (order $2$ vertex) such as one new edge results from the sum of the two former edges. The opposite edit is the {\em splitting}. Ghosting and splitting have no cost, therefore order $2$ vertices are \textit{de facto} irrelevant when computing the cost of an edit path;
    % These latter two edits cannot be performed on the root.
    c) Pruned tree edit distance between pruned dendrograms: pruning removes leaves with weights $\leq \varepsilon$, eventually aggregating homogeneous phenotypes. The operator $P_\varepsilon$ thus gradually discard intra-patient homogeneity, disclosing only the heterogeneous - independent - tumor phenotypes.
    Of note, $d_P^{\mu}$ is different from $d_E$ since the pruning modulates the effect of cardinality on the distance computation by removing redundant edges of the tree and compressing tree dimensionality.}
    \label{fig:distances} 
\end{figure}

After having obtained patient representation, we proceed to defining a distance between dendrograms, which can properly reflect the affinity between patients in terms of tree conformations as manifestation of intra-tumor heterogeneity. 
% A scheme of the formulation process is depicted in Fig. \ref{fig:distances} and detailed in the following.
A suitable metric should meet some requirements in order to produce effective results: (1) the comparison between dendrograms should reflect the properties of the point cloud they stem from: if two point clouds are close in terms of sparsity and conformation, we require the associated dendrograms to be close as well. In other words, any metric between dendrograms must hold some \textit{continuity} properties with respect to the original point clouds comparison;
(2) the metrics should weight differently the homogeneous part of the tree structures and the heterogeneous ones. This means that distance has to be evaluated as a trade-off between the extents of homogeneity and heterogeneity exhibited by the lesions of different patients. %: trees presenting a similar structure in terms of heterogeneous lesions but different in terms of homogeneous ones need to be closer than two trees exhibiting the opposite pattern. 

\subsubsection{Edit distance}
\label{sec:edit-dtstance}

Dendrograms are \emph{unlabelled} object which, in our context, may have a different number of leaves and do not hold any a-priori correspondence between the leaves in different objects.

The literature dealing with the comparison of dendrograms is reviewed in Appendix \ref{sec:distances}, where we detail the limitations that prevent us from employing existing distances in our context.
Recently, Pegoraro et al. \cite{pegoraro2021metric, pegoraro2021functional} proposed a novel distance for merge trees. 
Following the authors, we call this metric \emph{edit distance} for merge trees and indicate it with $d_E$. The metric $d_E$ is defined for weighted, rooted, unlabelled trees. 
As most of the metrics for unlabelled trees, its computational complexity has been shown to scale poorly with the number of leaves in the trees. However, it is particularly efficient for small-scale trees with respect to other metrics.
In our setting, trees present a number of leaves less or equal to the number of tumor lesions in a patient, that is a few dozens at most. Thus, we can run the comptuation of $d_E$ on general purpose machines, like personal computers.
Unlike other metrics, continuity properties are easily proven.
Moreover, $d_E$ is interpretable, easy to understand and to communicate.

As depicted in Fig. \ref{fig:distances}b), one tree $T$ can be modified and transformed into a different tree $T’$ by performing different sets of allowed modifications, each coming with its own cost (for details see Pegoraro et al. \cite{pegoraro2021metric}). 
The set of consequent edit operations which comes at the minimum cost is named the {\em optimal edit path} and represents the core of the edit distance between the two trees.
The distance $d_E$ is thus the total cost of the optimal edit path and is defined as:
\begin{equation}
    \centering
    d_E(T,T')=\inf_{\gamma\in\Gamma(T,T')} cost(\gamma)
\end{equation}
where $\Gamma(T,T')$ indicates all the possible edit paths which start in $T$ and ends in $T'$.
The algorithm for $d_E$ computation is exhaustively detailed in \cite{pegoraro2021metric}.
Through combinatorial objects called \emph{mappings}, it is shown that $d_E$ is a metric in the space of merge trees and that it can be computed with a Linear Integer Programming approach \cite{pegoraro2021metric}.

Upon these premises, we proceed to verify the two aforementioned conditions. Specifically, we prove the continuity property of $d_E$ (1) and propose a modification of $d_E$ as to meet the homogeneity-heterogeneity requirement (2).

\subsubsection{Continuity property of $d_E$} 
\label{sec:continuity}

As previously stated, the distance between dendrograms must hold continuity results with respect to the original point clouds comparison: under certain hypotheses, if two clouds are pointwise close, also their merge trees should be close with respect to $d_E$.
In Fig. \ref{fig:distances}a), we introduce the Hausdorff metric between point clouds (for formal definition see Appendix \ref{sec:proof-continuity}). It can be interpreted as a measure of the pointwise proximity between two point clouds and provide a comparison between the heterogeneity of two patients' diseases. In Appendix \ref{sec:proof-continuity}, we prove that Hausdorff-closeness for point clouds implies Edit-closeness for the associated dendrogram objects, i.e., multi-lesion patients representation. 

\subsubsection{Homogeneity-heterogeneity trade-off} 

In the edit distance $d_E$, the distance values are strongly dependent on the clouds cardinalities, meaning that pairs of point clouds with higher cardinalities tend to be farther apart from pairs of point clouds with smaller cardinalities. At first sight, such assumption sounds reasonable for stratification purposes. In fact, patients with multiple lesions are known to exhibit a more severe disease than patients with fewer lesions, as the spreading of the tumor entails prognostic power.
Still, the mere counting of lesions lacks of robustness in perspective studies and, in this context, may overshadow the variability between hierarchical dendrograms induced by intra-patient heterogeneity.
For this reason, we propose a modification of the metric $d_E$ as to mitigate cardinality issue.

\subsubsection{Pruned edit distance}
\label{sec:pruned_dist}

The kind of variability we are interested in is the one induced by patient-wise heterogeneity between lesions. By construction of the dendrogram representation, two lesions of a patient are heterogeneous - in terms of radiomic/imaging description - according to the length of the dendrogram branches connecting them. The longer the branches, the higher the inter-lesions heterogeneity and, viceversa, the shorter the branches the more homogeneous the patient’s disease phenotypes. Accordingly, we may want to modulate the extent to which we consider edit costs according to branch length.
% As stated in Section \ref{pat_representation}, we expect homogeneity between lesions to be reflected into proximity between points in a point cloud, which reflects into small edges in the associated dendrogram. 
In particular, we may want to induce edits applied on small edges to contribute less to the final distance than bigger edges, which we deem more relevant for stratification purposes.

\begin{figure}[h!]
	\centering
	\includegraphics[width = \textwidth]{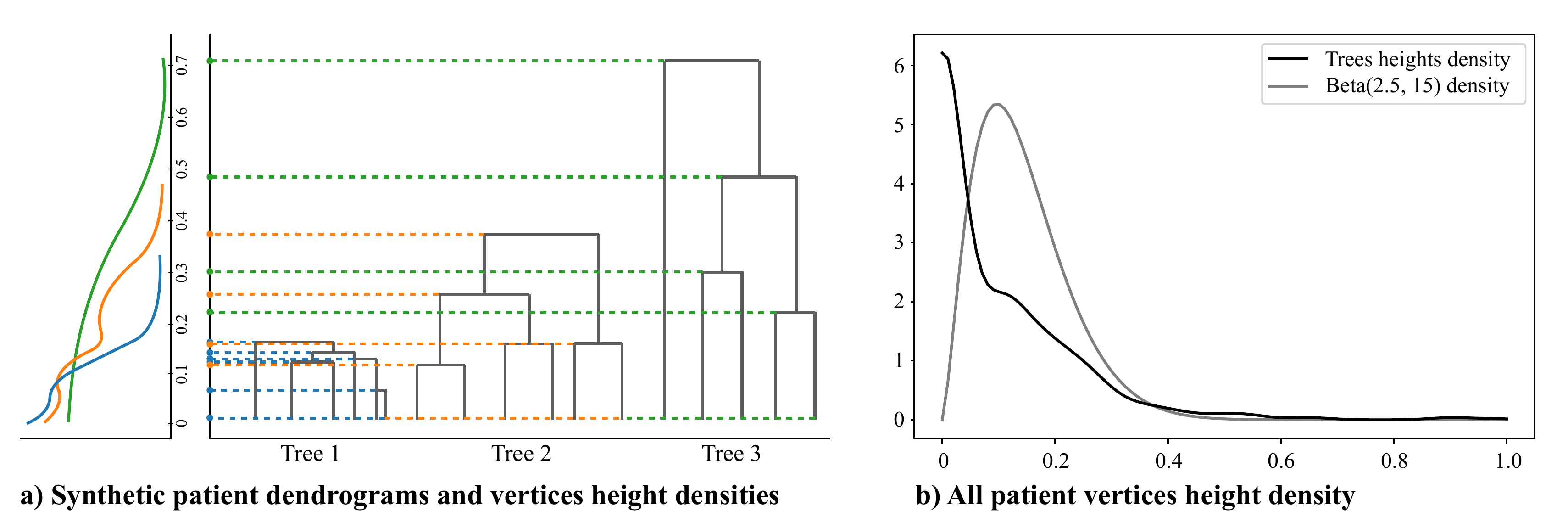}
	\caption{Choice of $\mu$: a) costruction of qualitative densities of the vertices heights in three example dendrograms: the velocity with which leaves get merged in a dendrogram, i.e., edges length variability, reflects the heterogeneity characterization of lesions. Per every dendrogram, branches heights (rescaled on $[0,1]$ dividing by the highest value)  are annoted on the left and their associated density is inspected. The vertices heights of a patient exhibiting homogeneous lesions concentrates in a small real interval $[0,a]$ - with $a>0$ (blue tree); the vertices heights of a patient exhibiting heterogeneous lesions spread in a range of values far from zero $[a,b]$, with $a,b>0$ (green tree); a patient showing groups of homogeneous lesions, the one heterogeneous to the others, is associated to a dendrogram with an explicit clustering structure with clusters with multiple close leaves (orange tree). The vertices heights distribution displays two components, reflecting both the homogeneity of similar lesions - with values close to $0$ - and the heterogeneity of dissimilar clusters - with values far from $0$;
	b) $\mu$ provides the coefficients with which to weight the different pruning cutoffs $\varepsilon$, to neglect the homogeneity within clusters of similar lesions’ phenotypes and bring out the informative heterogeneity between different phenotypes. To efficient the computation, a parametric shape of $\mu$ is used and empirical heights distributions of all patients (black line) is exploited to model the distribution. In the population heights distribution, we discern both homogeneous and heterogeneous phenotypes. The two components are demarked with a saddle point on $0.15$. Accordingly, low weights of $\mu$ should be associated to $\varepsilon \ll 0.15$ and $\varepsilon \gg 0.15$ and high weights to $\varepsilon \simeq 0.15$. 
	In fact, low $\varepsilon$ values entail pure homogeneity information while high $\varepsilon$ values would lead to discarding useful heterogeneity information. We thus infer to model $\mu$ as an asymmetric bell-shaped density function with one peak centered in the saddle point of the heights distribution. The Beta family of distributions, supported in $[0,1]$, well meets the requirements; it simplifies both the numeric integration procedure and the results' interpretation. The Beta-shaped $\mu$ is centered on $0.15$ (grey line), properly tuning $\alpha$ and $\beta$ shape parameters ($\alpha=2.5, \beta=15$).
    }
    \label{fig:mu}
\end{figure}

We introduce the pruning operator $P_\varepsilon$ as regularization strategy, which deletes leaves associated with edges whose weights are so small that one may want to neglect them in the analysis of heterogeneity. 
Given a threshold $\varepsilon$, we consider for deletion all leaves whose father-child edge has weight $\leq \varepsilon$.
However, when two or more of candidate leaves share the same father, i.e. they are \emph{siblings}, we delete all the leaves but the one with the bigger weight. 
Moreover, if the weights of the siblings are equal, as it is often the case in clustering dendrograms, we randomly choose to keep one of them, delete the other(s) and, eventually, \textit{ghost} their father (see Fig. \ref{fig:distances} for meaning of ghosting). 
% Moreover, whenever one leaf has no siblings with weight below the threshold, it is deleted and the father is ghosted if it becomes an order $2$ vertex. After a first iteration of deletions, the resulted tree may still have leaves with weights below $\varepsilon$ and thus this procedure is applied recursively until no leaves with small edges can be found.  
This pruning operation is recursively iterated until no leaves with small edges can be found.
To note, removing only one leaf in case of two small-weight siblings is equivalent to considering the two leaves as clustered together from the ``beginning'' in the hierarchical clustering procedure. Accordingly, siblings leaves (lesions) entail phenotype expressions so similar to be considered as one single imaging phenotype.
In this way, the pruned tree displays the number of \textit{different} phenotypes coexisting in the patient instead of the mere number of lesions. 
Fig. \ref{fig:distances}c) displays the edits needed for transforming a pruned tree into another, whose costs determine the pruned edit distance.

Operationally speaking, the ``correct'' value of $\varepsilon$ is a-priori unknown and needs to be tuned with a complexity-information trade-off. To enhance the robustness of this parameter choice, we take the weighted average of the distances between two trees pruned with all the possible values of $\varepsilon$.
Accordingly, the definition of \emph{pruned edit distance} for general merge trees develops as follows. Given two merge trees $T$ and $T'$, the pruned edit distance is:
\begin{equation}
    \centering
    d^\mu_P(T,T'):=\int_{\mathbb{R}} d_E(P_\varepsilon(T),P_\varepsilon(T'))d\mu(\varepsilon)= \mathbb{E}_{\varepsilon\sim\mu}[ d_E(P_\varepsilon(T),P_\varepsilon(T')]
\end{equation}
where $\mu$ is a finite measure on $\mathbb{R}$ which provides the weighting strategy across different values of $\varepsilon$ in order to compute a weighted average among trees distances. The higher the mass $\mu$ associated to an interval $[a,b]$, the bigger the contribution to the final result of the tree distance according to $\varepsilon \in [a,b]$.
In other words, the measure $\mu$ allows to control the contribution to the final distance of branches with weight below $\varepsilon$, which are indeed homogeneous enough to be removed.  
Fig. \ref{fig:mu} elucidates the choice of $\mu$ tuned on case study data.
Note that if we have a sequence of weakly converging probability measures $\mu_n \rightharpoonup \mu $, then $d^{\mu_n}_P(T,T') \rightarrow d^{\mu}_P(T,T')$. This implies that the proposed distance is robust with respect to the choice of $\mu$: similar measures $\mu$ (in the sense of weak convergence) would give similar distances.

To assess the different behaviours between $d_E$ and $d^\mu_P$ and the extent to which $d_P^\mu$ is suitable for our purposes, in Appendix \ref{sec:d_mu_simulation} we present a detailed simulation study.
Moreover, we can prove that, under general conditions on $\mu$, $d_P^\mu$ is still a metric (for proof see Appendix \ref{sec:proof-metric}). 

\newpage

\noindent \textbf{APPENDIX}

\appendix

\section{Patients' personal information summary}
\label{patient_imformation}

Tables \ref{tab:summary_continuous} and \ref{tab:summary_categorical} summarize the patients' population.

\begin{table}[h!]
\centering
    \begin{tabular}{lllll}
    \toprule 
    Variable & Mean &  Std. dev. & Median & Range \\
    \midrule
    \rowcolor{black!10} Age & 72.09  & 7.03   & 71.68 & 54.88 – 85.24 \\
    Total volume  &  16.41 &  34.72   & 3.16  &  0.22 – 207.70 \\
    \rowcolor{black!10} Gleason Score & 7.73 &  1.03  & 7.00 & 5.00 – 9.00 \\
    PSA & 18.16 &  70.96   & 2.66  &  0.09 – 591.00 \\
    \bottomrule
    \end{tabular}
\caption{Statistical summary of patients' personal information (continuous variables).}
\label{tab:summary_continuous}
\end{table}

\begin{table}[h!]
\centering
    \begin{tabular}{lll}
    \toprule 
    Variable & & Number of patients (\%) \\
    \midrule
    \rowcolor{black!10} Number of metastases  &  Oligo ($<$3) & 38 (41.3\%) \\
    \rowcolor{black!10} &  Multi ($\geq$3) & 54 (58.7\%) \\ \hline
    \rowcolor{black!10} &  Oligo ($<$5)  & 60 (65.22\%) \\
    \rowcolor{black!10} &  Multi ($\geq$5) & 32 (34.78\%) \\ \hline
    \rowcolor{black!10} &  Intermediate (3$\leq$n$<$5) &  22 (23.92\%) \\
    Gleason Score (dichotomous) & $<$7 & 8 (8.7\%) \\
    &  $=$7 & 45 (48.91\%) \\
    &  $>$7 & 31 (33.69\%) \\
    &  missing & 8 (8.7\%) \\
    \rowcolor{black!10} Ongoing therapy & Y  & 33 (35.87\%) \\
    \rowcolor{black!10} &  N  & 59 (64.13\%) \\
    Initial therapy & RP & 23 (25\%) \\
    &  RP+RT &  52 (56.52\%) \\
    &  RT & 9 (9.78\%) \\
    &  missing & 8 (8.7\%) \\
    \rowcolor{black!10} PSA (dichotomous) &  $\leq$ 1.93 & 33 (35.87\%) \\
    \rowcolor{black!10} &  $>$1.93 &  48 (52.17\%) \\
    \rowcolor{black!10} &  missing & 11 (11.96\%) \\
    \bottomrule
    \end{tabular}
\caption{Statistical summary of patients' personal information (categorical variables).}
\label{tab:summary_categorical} 
\end{table}

\section{Distance metrics for trees: literature review}
\label{metrics_distance}

\subsection{Different Kinds of Trees}
\label{kindsoftrees}
Before reviewing the existing metrics for distances among tree objects, it is worth to list the different kinds of tree that have been defined throughout the years. We integrate the different definitions and approaches presented in literature under the light of this work’s objectives, by adopting the most appropriate definition for our purposes.

We start from the general definition of a \emph{tree structure} found in \cite{pegoraro2021metric} and \cite{pegoraro2021functional}. 

\begin{defi}
A tree structure $T$ is given by a set of vertices $V_T$ and a set of edges $E_T\subset V_T\times V_T$ which form a connected rooted acyclic graph.  We indicate the root of the tree with $r_T$. We say that \(T\) is finite if \(V_T\) is finite. The order of a vertex of \(T\) is the number of edges which have that vertex as one of the extremes. 
Any vertex with an edge connecting it to the root is its child and the root is its father: this is the first step of a recursion which defines the father-children relationship for all vertices in \(V_T.\)
% In this way we recursively define father and children (possibly none) for any vertex on the tree. 
The vertices with no children are called leaves or \textit{taxa}. The set of leaves is called $L_T$. Vertices which are not leaves are called internal and they are collected in the set $I_T$. The relation $father > child $ induces a partial order on $V_T$. The edges in $E_T$ are identified in the form of ordered couples $(a,b)$ with $a<b$.
A subtree of a vertex $v$ is the tree structure whose set of vertices is $\{x \in V_T| x\leq v\}$. 
\end{defi}

On top of this definition, tree structures primarily divide into unlabelled trees - also called tree-shapes - and labelled trees, depending on whether the set-related information contained in the leaves - also called labels - is considered or discarded. When dealing with unlabelled trees we want to work up to the following set of maps. 

\begin{defi}
Two tree structures $T$ and $T'$ are isomorphic if there exists a monotone bijection $\eta:V_T\rightarrow V_{T'}$ inducing a bijection between the edges sets $E_T$ and $E_{T'}$: $(v,v')\mapsto (\eta(v),\eta(v'))$. Such $\eta$ is an isomorphism of tree structures. 
\end{defi}

On the other hand, labelled trees are of interest in many applications, where are used to infer information about labels’ description. Thus, the comparison between trees has to be driven with regard to the labeled structure.

\begin{defi}
A label-preserving isomorphism between the tree structures $T$ and $T'$ is an isomoprhism of trees $\eta:V_T\rightarrow V_{T'}$ such that $\eta_{L_T}=id_{L_T}$. The term $id_A$ is the identity map on a set $A$. 
\end{defi}

Accordingly, we provide the following definitions.

\begin{defi}
An unlabelled tree or, equivalently, a tree shape, is the isomorphism class of a tree structure. 
A labeled tree is the label-preserving isomorphism class of a tree structure. Labelled phylogenetic trees and labelled hierarchical clustering dendrograms are names which are used instead of labeled trees in some precise scientific contexts.
\end{defi}

Beside unlabelled and labelled trees, there are some in-between structures which posses some additional ordering properties on the vertices. In particular: 

\begin{defi}
A ranked tree shape is a tree shape $T$ with a complete ordering of the internal vertices $I_T$. Similarly we may have ranked labeled trees. 
\end{defi}

A step forward in the analysis of tree objects is represented by weighted trees (or clustering dendrograms). Such structures entail information about both the tree structure and the length of the branches, which may carry some relevant insights for many applications.

\begin{defi}
A weighted tree shape is a tree shape $T$ along with a weight function $w_T:E_T\rightarrow \mathbb{R}_{>0}$. The weight value of a branch/edge is sometimes called length of the branch, due to its positive value. In some contexts such trees are also called genealogies.
\end{defi}

Unlabelled clustering dendrograms are a particular case of weighted tree shapes. To introduce them we need to formalize the following notation.  
Given a tree structure $T$ ad a vertex $v\in V_T$, we call $\zeta_v$ the set $\zeta_v=\{v'\in V_T|v\leq v' \leq r_T \}$. That is, $\zeta_v$ contains all the points between $v$ and the root $r_T$.

\begin{defi}
A weighted tree shape $T$ is isochronously sampled - or, equivalently, is a clustering dendrogram - if for any couple of leaves $(l$, $l')$ we have $\sum_{v\in\zeta_l}w_T(v)=\sum_{v'\in\zeta_{l'}}w_T(v')$. This means that the leaves are all at the same distance from the root.
If this does not happen, the tree shape is said to be heterochronously sampled. 
\end{defi}

In this paper we focus our interest on isochronously sampled weighted tree shapes, since dendrograms obtained from hierarchical clustering are indeed  isochronously sampled. For this reason, we use the word (hierarchical clustering) dendrogram as to identify an isochronously sampled weighted tree shape.
Nevertheless, the result in Section \ref{sec:proof-metric} holds also for heterochronously sampled trees.

\subsection{Distance Metrics between Trees}
\label{sec:distances}
As previously stated, dendrograms are \emph{unlabelled} object which, in our context, may have a different number of leaves and do not hold any a-priori correspondence between the leaves in different objects.
The literature dealing with the comparison of dendrograms divide in two macro-areas, including (1) metrics defined for clustering dendrograms and (2) metrics designed for \emph{merge trees}.
The first family of metrics mainly deals with \emph{labeled} trees as byproducts of a hierarchical clustering algorithm. We refer to Flesia et al. \cite{flesia2009unsupervised} for an exhaustive review of distance definitions. This kind of metrics are known to be heavily dependent on the graph structure of the dendrograms, leading to some limitations when comparing dendrograms with a different number of leaves. Moreover, theoretical continuity results with respect to dendrogram-associated point clouds are often lacking.
On the other hand, within topological data analysis, dendrograms are ofen referred as a particular case of merge trees, obtained when all the leaves of a merge tree lie at height $0$. This allows to transfer merge tree literature, the second family of metrics, to dendrogram analysis. 
In this Section we review the first family of metrics, detailing definitions and limitations of employing those distances in our context. Most of the metrics belonging to the second family, instead, shares one main drawback, namely the out of reach computational cost \cite{merge, bauer2020reeb, cardona2021universal}, which makes them unsuitable for our application. 
Besides, the metrics with more performing algorithms \cite{pont2021wasserstein, merge_farlocca} still lack the theoretical investigation to assess some practical properties, making them less worthy than others.

% In this Section we review some of the metrics that have been proposed in literature to compare sets of tree objects as defined in the previous Section. We thus briefly introduce such metrics, specifying for which kind of trees they are suitable and discussing their properties in relation to our context.

\subsubsection*{LAB}
One of the main points of interest in comparing trees is to interpret them as explaining the evolution of a fixed set of labels under some ``agglomerative'' criterion, being it a clustering criterion or a genetic evolution summary. For this reason, a lot of research focused on comparing labelled trees. The most notable examples of metrics for weighted labelled trees are the Robinson-Foulds metric \citep{robinson1979comparison} and the BHV metric \citep{billera2001geometry}.
A number of limitations of these metrics has been pointed out by \cite{smith2019bayesian} and \cite{smith2020information}. In particular, severe shortcomings prevent researchers from comparing weighted tree shapes with a variable number of leaves.

\subsubsection*{SHAPE}

Recently \cite{colijn2018metric} proposed a distance to compare tree shapes. Such metric is based on a numeric representation of tree shapes, obtained with a labeling related to the tree isomorphism algorithm. Then it produces vectors enriching this numeric information with indices based on frequencies of subtrees shape and other statistical summaries of the tree, including length-related information like spectral differences, Sackin or Colless imbalance, etc. The metric between trees is obtained as the Euclidean metric between these vectors. A key point for us is that the contributions of the part of the vector depending on the tree shape and the one obtained from the length of the edges are independent.
Accordingly, although this metrics shows good linearity and convexity properties, it reveals too sensitive to the underlying tree structure.

\subsubsection*{MAT}

\cite{kim2020distance} produces a metric to compare ranked tree shapes and ranked genealogies based on a matrix representation of ranked tree shapes. The dimension of such matrix is determined by the number of leaves possessed by the tree and thus to be coherently compared, two trees must possess the same number of leaves. In fact, the metric between trees is defined as some Euclidean metric between the corresponding matrices. This limitation makes this metric unsuitable for our purposes.

\subsubsection*{LAP}

A graph-oriented approach is pursued by \cite{lewitus2016characterizing}: the authors represents a tree shape (possibly weighted and heterochronously sampled) by means of its graph laplacian matrix. From such matrix the sequence of eigenvalues is extracted: eigenvalues are know to be heavily dependent on the graph connectivity and, in particular, on shortest-path lengths between vertices. Specifically, high eigenvalues arise from areas of the graph which have sparse nodes with long branches, low eigenvalues correspond to very dense regions of the graph (many nodes connected by short edges). To get a more versatile summary of the tree, this sequence is then smoothed and normalized, to obtain a spectral density profile. In order to compare tree shapes, such densities are compared. One drawback of this representation is that the ``operator'' which maps a tree shape into a density, has no guarantees to be injective. Moreover, any information about the rooted nature of the tree and the ordering structure of leaves is discarded, leading to poor results.

\subsubsection*{KER}

\cite{poon2013mapping} presented a kernel approach to measure similarities between tree shapes. The comparison proceeds as follows: the kernel looks for all possible \emph{subset trees} -  contiguous portions of (unweighted) subtrees - which are shared between the trees the kernel and adds up positive contributions for every shared subset tree, weighted by similarity between lengths. As a result, higher scores will be assigned to trees which share, locally, similar structures and with similar weights. 
However, it is to be noted that this is just a measure of similarity and does not provide a proper distance between trees. Moreover, the authors do not present a comparison with other tree metrics or similarities, nor enough information to grasp for which purpose their similarity measure is best suited, which kind of variability between trees it tends to capture and which possible pathologies presents. The authors state that their ``approach is similar to the Robinson-Foulds metric'' and thus it may suffer some of the severe shortcomings pointed out for such metric \citep{smith2019bayesian, smith2020information}.

\bigskip

An immediate observation that we can make is that most of the aforementioned metrics are not suitable for our purpose: we need to compare weighted tree shapes with possibly a different number of leaves. 
Apart from [SHAPE], [LAP] and [KER] all the others are discarded.
Note that any kind of metric defined for labelled trees can be extended to work with (weighted) tree shapes by trying all possible permutations of labels; but this approach is clearly computationally out of reach even for small trees.

There are also some reasons for which we discard the metric [LAP]. First, clustering dendrograms are intrinsically rooted objects, thus there is a well defined height where all the objects are clustered together. From another points of view, clusters enjoy a partial order relationship given by inclusion which is reflected by the rooted nature of hierarchical dendrograms. The [LAP] approach, on top of not being a proper distance, completely throws away this information. 

The metrics [KER] and [SHAPE], in addition to the shortcomings already pointed out in the previous lines, share the following drawback: they are very sensitive to the underlying tree-shapes. For instance, the value of the metric [SHAPE] cannot be arbitrarily close to zero if two tree shapes are different. 

\section{Building the vertices heights curve}
\label{sec:building_curves}

In Figure \ref{fig:trees_curves}, we explain in details how to build the curve displaying the filtered heights of one tree’ vertices.
In the plot there are three detailed examples showing how to obtain the curves in Figure 2a. The colors of the curves match the colors of the trees. For each tree and for every $h\in \mathbb{R}$,  the value $f(h)$ is equal to the number of vertices (highlighted with dots) which lie above $h$. For instance consider the blue dendrogram: we obtain a value of $3$ until we reach the height where two leaves merge, and then $1$ for a small interval of values, and lastly $0$ from the height of the root of the blue tree onwards.

\begin{figure}[h!]
\centering
\includegraphics[width=0.7\textwidth]{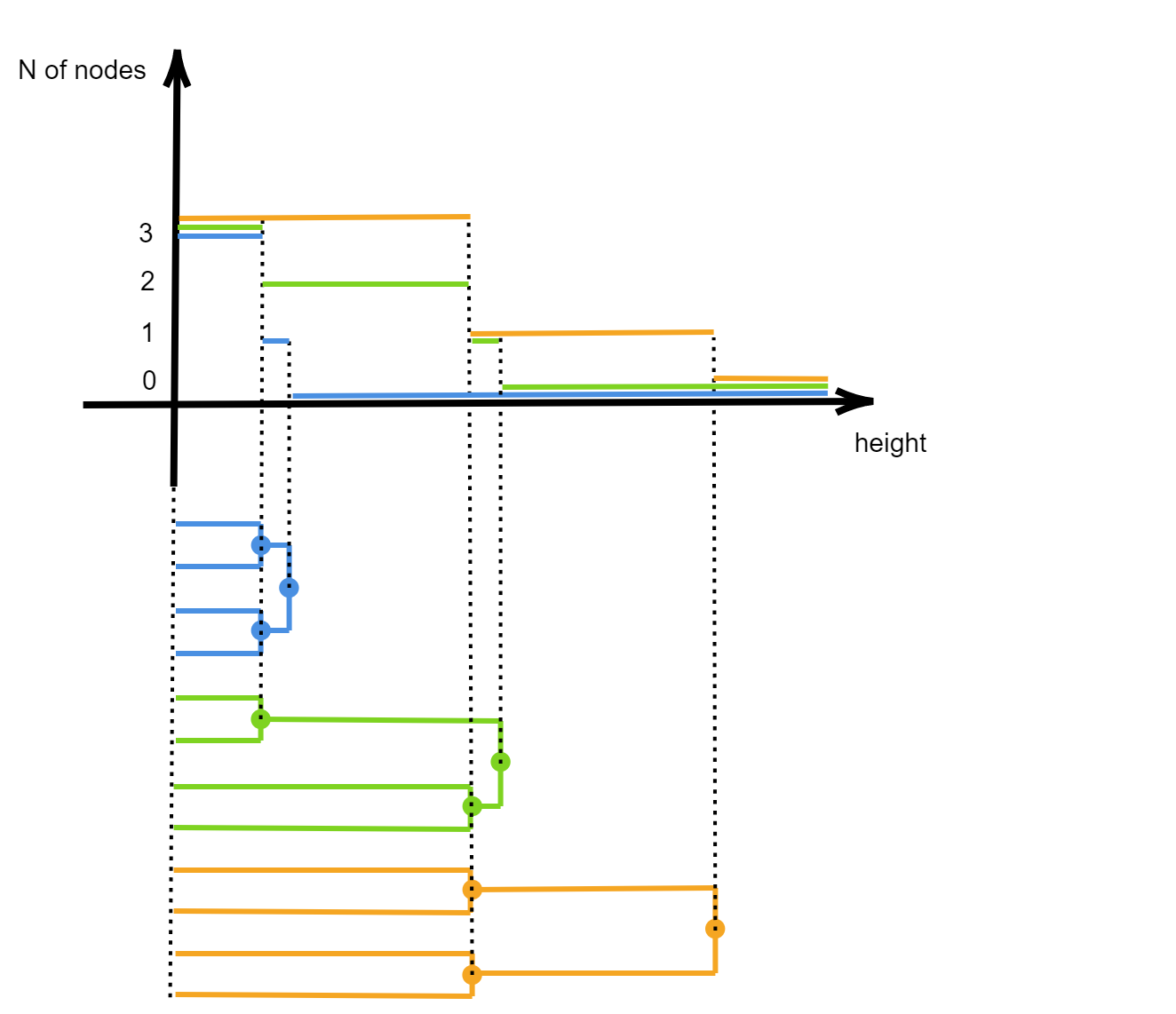}
\caption{Procedure for building the vertices heights curve of three example patient-tree.}
\label{fig:trees_curves} 
\end{figure}

\section{Dendrograms construction}
\label{sec:build_dendro}

In this section we present few technical definition that we need in order to describe the dendrogram representation we employ. We describe the procedure in the general case of having a finite metric space $(X,d)$ i.e. a finite set $\{x_1,\ldots,x_n\}$ with a metric $d:X\times X\rightarrow \mathbb{R}_{\geq 0}$ which is reflexive, symmetric and satisfies the triangular inequality.
In our case we work with $\{x_1,\ldots,x_n\}\subset \mathbb{R}^n$ and the Euclidean norm.

\begin{defi}
A tree structure $T$ is given by a finite set of vertices $V_T$ and set of edges $E_T\subset V_T\times V_T$ which form a connected rooted acyclic graph. The order of a vertex is the number of edges which have that vertex as one of the extremes. 
Any vertex with an edge connecting it to the root is its child and the root is its father.
In this way we recursively define father and children (possibly none) relationships for any vertex on the tree. The vertices with no children are called leaves and are collected in the set $L_T$, while the set of children of a vertex $x\in V_T$ is called $child(x)$. Similarly, the vertex $father(x)$ is the father of the vertex $x$. 

The relationship $father > child $ induces a partial order on $V_T$. The edges $E_T$ are given in the form of ordered couples $(a,b)$ with $a<b$.
For any vertex $v\in V_T$, $sub_T(v)$ is the subtree of $T$ rooted in $v$, that is the tree structure given by the set of vertices $v'\leq v$. If clear from the context we might omit the subscript $T$. 
\end{defi}

\noindent Now, to obtain a dendrogram we need to add some kind of length measure to a tree structure.

\begin{defi}
A merge tree $(T,f)$ is a finite tree structure T coupled with a monotone increasing function (with respect to partial ordering on $V_T$) \(f:V_T\rightarrow \mathbb{R}\). If $f(l)=0$ for all $l\in L_T$, then we say that the merge tree is a dendrogram. The function $f$ also defines a weight value for every edge $e=(v,father(v))$: $w_T(e)=f(father(v))-f(v)$.
\end{defi}

\noindent To build a hierarchical clustering dendrogram $T_C$ from a finite metric space $(C,d_C)$ we proceed as follows. With $K$ we indicate the set of clusters we are considering:
\begin{itemize}
    \item[(S0)] at the beginning $K=\{\{c\}\mid c\in C\}$, and every $c\in C$ is associated to a leaf $v_c\in V_{T_C}$ with $f(v_c)=0$;
    \item[(S1)] consider all the couples of clusters  $k_1,k_2\in K$ and we measure the distance $d(k_1,k_2)$ according to some \emph{linkage};
    \item[(S2)] pick $k,k'\in K$ such that $d(k,k')=\min_{k_i\in K;k_1\neq k_2} d(k_1,k_2)$ and add the vertex $v_{kk'}$ to $V_{T_C}$ with $f(v_{kk'})=d(k,k')$. Then remove $k$ and $k'$ from  $K$ and add $k\cup k'$ to $K$;
    \item[(S3)] start again from (S1) unless $K=C$.
\end{itemize}

\noindent The linkage determines the distance $d(k_1,k_2)$ between $k_1,k_2\subset C$ and the most common examples are:
\begin{itemize}
    \item single linkage: $d(k_1,k_2)=\min_{c_i\in k_i}d_C (c_1,c_2)$
    \item complete linkage: $d(k_1,k_2)=\max_{c_i\in k_i}d_C( c_1,c_2)$
    \item average linkage: $d(k_1,k_2)=(\#k_1 \cdot \#k_2)^{-1}\cdot\sum_{c_i\in k_i}d_C(c_1,c_2)$, where $\#k_i$ is the cardinality of the finite set $k_i$.
    \item ward linkage: see  \cite{ward1964linkage}
\end{itemize}

\noindent It is well known that single linkage is very sensitive to outliers, while complete linkage is the most conservative choice in term of clustering points together. Average linkage displays a kind of in-between behaviour. For this reason we resorted to average linkage.

\section{Continuity Proposition}
\label{sec:proof-continuity}
Having a continuity result of the distance between dendrograms with respect to some metric between point clouds would surely benefit the consistency and the interpretability of the framework: whenever a representation of a datum is employed, looking at how changes in the representation reflect on changes in the initial datum can help both assessing the consistency of the pipeline, and familiarizing with the representation. 
In our case we are particularly interested in looking at what happens at the distance between dendrograms as two sequences of point clouds get closer and closer. To do so, we introduce the definitions of the Hausdorff and Gromov-Hausdorff metrics between point clouds. We then prove the continuity proposition between Hausdorff distance and the Edit distance, from which it follows the continuity proposition between Gromov-Hausdorff distance and the Edit distance.

%The distance between dendrograms must hold continuity results with respect to the original point clouds comparison. 
%To prove so, we introduce the definitions of the Hausdorff and Gromov-Hausdorff metrics between point clouds. We then prove the continuity proposition between Hausdorff distance and the Edit distance, from which it follows the continuity proposition between Gromov-Hausdorff distance and the Edit distance.
\smallskip

\noindent Given  $C=\{x_1,\ldots, x_n\}$ and $C'=\{y_1,\ldots, y_m\}$ two point clouds in a metric space $(X,d)$, we can build at least a function $\gamma:C\rightarrow C'$ such that $\gamma(x_i)$ is (one of) the closest point(s) to $x_i$, belonging to the cloud $C'$. Similarly, we can build $\varphi:C'\rightarrow C$ so that $\varphi(y_j)$ is (one of) the closest point(s) to $y_j$, belonging to the cloud $C$.
The Hausdorff distance between $C$ and $C'$ is given by:
\begin{equation}
    \centering
    d_H(C,C')=\text{max}\{\text{max}_{x\in C}d(x,\gamma(x)),\text{ max}_{y\in C'}d(y,\varphi(y))\}
\end{equation}

\noindent The distance $d_H$ has been proven to be a metric for the space of all compact subsets of $X$ \cite{rockafellar2009variational}.

\smallskip

\noindent Leveraging on the Hausdorff distance, given two compact metric spaces $X$ and $Y$ we define the Gromov-Hausdorff metric as $d_{G-H}(X,Y):=\inf d_H(\gamma(X),\varphi(Y))$ where $\gamma$ and $\varphi$ vary over all possible isometries of (respectively) $X$ and $Y$ into another (common) metric space $Z$ \cite{burago2022course}.

\smallskip

\noindent Consider two point clouds in the metric space $(X,d)$, $C=\{x_1,\ldots, x_n\}$ and $C'=\{y_1,\ldots, y_m\}$ and we consider $T_C$ and $T_{C'}$ the single linkage hierarchical clustering dendrograms obtained from $C$ and $C'$ respectively. 
In the following, we prove the following result:

\begin{prop}\label{prop:cont_functions}
Given $C=\{x_1,\ldots, x_n\}$ and $C'=\{y_1,\ldots, y_m\}$  point clouds in a metric space $(X,d)$ and given $T_C$ and $T_{C'}$ single linkage hierarchical clustering dendrograms obtained from $C$ and $C'$ respectively, there is a simplicial complex $S$
and two functions $f:S\rightarrow \mathbb{R}$ and $g:S\rightarrow \mathbb{R}$ such that the merge tree associated to $f$ (via sublevel set filtration) is isomorphic to $T_C$, 
the merge tree associated to $g$ is isomorphic to $T_{C'}$, and $\parallel f-g\parallel_\infty \leq 2d_H(C,C')$.
\end{prop}

\begin{proof}
Let $\gamma:C\rightarrow C'$ and $\varphi:C'\rightarrow C$ be the two operators which map a point of a point cloud $C'$ to (one of) the closest point(s) of the other cloud $C$ and viceversa.

\smallskip

\noindent Consider the following simplicial complex $S$. Its $0$ simplices are
${x_1,\ldots,x_n,y_1,\ldots, y_m}$ and its $1$ simplices are all possible edges between $0$ simplices, forming a complete graph. 

\smallskip

\noindent Now we define two functions $f:S\rightarrow \mathbb{R}$ and $g:S\rightarrow \mathbb{R}$ such that
the merge trees $T_f$ and $T_g$ obtained with the lower star filtration from $f$ and $g$ (see \cite{pegoraro2021functional}, Section 2) are isomorphic to $T_C$ and $T_{C'}$. 

\smallskip

\noindent Define: $f(s)=0$ for every $0$ simplex $s$. Then for a $1$ simplex of the form $e_{ij}=(x_i,x_j)$, we have $f(e_{ij})=d(x_i,x_j)$. For $1$ simplices of the form $e'_{ij}=(y_i,x_j)$ we have $f(e'_{ij}) = d(\varphi(y_i),x_j)$. Lastly, for $1$ simplices of the form  $e''_{ij}=(y_i,y_j)$ we have $f(e''_{ij}) = d(\varphi(y_i),\varphi(y_j))$.
Note that $t\in Im(f)$ iff $t=d(x_i,x_j)$ for some $i$ and $j$. Clearly, $f$ is a finite set and we can order it: $t_0=0<t_1<\ldots$.

\smallskip

\noindent Similarly we define $g(e''_{ij})=d(y_i,y_j)$, $g(e'_{ij}) = (y_i,\gamma(x_j))$ and $f(e''_{ij}) = (y_i,y_j)$.

\smallskip

\noindent Consider now the connected components of the graph $S^f_{t}:=\{s\in S|f(s)\leq t\}$ for $t\in \mathbb{R}$. If $t<0$, $S^f_{t}$ is empty. If $t=t_0=0$, then all $0$ simplices are in $S^f_0$, plus the $1$ simplices of the form $(x_i,y_j)$ such that $\varphi(y_j)=x_i$ and $(y_i,y_j)$ such that $\varphi(y_i)=\varphi(y_j)$.
This means that every vertex $y_i$ is connected with exactly one point $x_j$ and with all other $y_k$ such that $\varphi(y_k)=x_j$. That is, there are $n$ path connected components, one for each $x_i$. Call such components $[x_i]$.

\medskip

\noindent Consider the value $t=t_1=d(x_i,x_j)$. For every $y\in\varphi^{-1}(x_i)$ and $y'\in\varphi^{-1}(x_j)$, we have $f((y,y'))=f((x_i,y'))=f((y,x_j))=f((x_i,x_j))=d(x_i,x_j)$ and so all these $1$ simplices get added, when passing from $S^f_0$ to $S^f_{t_1}$.
Moreover, these are the only ones which get added. 
Which means that we get all possible edges between $[x_i]$ and $[x_j]$ but all others components are left unchanged. And this happens whenever we hit a level $t_k=d(x_i,x_j)$: we add to the simplicial complexes $S^f_{t_k}$ all possible edges between $[x_i]$ and $[x_j]$.

\smallskip

\noindent Now, we build the single linkage hierarchical dendogram $T_C$ associated to $C$, with labels given by 
$\{\{x_1\},\ldots,\{x_n\}\}$, and the merge tree $T_f$ associated to $f:S\rightarrow \mathbb{R}$ with labels $\{[x_1],\ldots, [x_n]\}$. 
An internal vertices of $T_{C}$ indicating the merging of two leaves $\{x_i\}$ and $\{x_j\}$ will be called $\{x_i,x_j\}$, and similarly
a vertex called $\{x_i,x_j,x_k\}$ indicates that the leaves of the subtree rooted in that vertex are $\{x_i\}$, $\{x_j\}$ and $\{x_k\}$. In the same fashion, an internal vertex of $T_f$ where to components $[x_i]$ and $[x_j]$ merge is named $[x_i]\bigcup[x_j]$. A vertex called $[x_i]\bigcup[x_j]\bigcup[x_k]$ is associated to the origin of the connected component $[x_i]\bigcup[x_j]\bigcup[x_k]$.
Thus, we can define a map $\eta:V_{T_{C}}\rightarrow V_{T_f}$ induced by $\eta({x_i})=[x_i]$ and $\eta(\{x_i,x_j,x_k\})=[x_i]\bigcup[x_j]\bigcup[x_k]$ which is an isomorphism of merge trees.
An analogous proof yields the isomorphism between $T_{C'}$ and $T_g$.

\smallskip

\noindent To conclude the proof it is enough to notice that: $||f-g||_{\infty} \leq 2\varepsilon$ with $\varepsilon = d_H(C,C')$. In fact, for vertices $s$: $f(s)=g(s)=0$. For an edge $e$, we have the following possibilities:
\begin{itemize}
    \item $e=(x_i,x_j)$: $|f(e)-g(e)|=|d(x_i,x_j)-d(\gamma(x_i),\gamma(x_j))|$. We have $d(x_i,\gamma(x_i))\leq \varepsilon$, $d(x_i,x_j)\leq d(\gamma(x_i),\gamma(x_j)) + 2\varepsilon$ and $ d(\gamma(x_i),\gamma(x_j)) \leq d(x_i,x_j) + 2\varepsilon$; which, together, give $|f(e)-g(e)|\leq 2 \varepsilon$. 
    \item $e=(y_i,y_j)$: $|f(e)-g(e)|=|d(\varphi(y_i),\varphi(y_j))-d(y_i,y_j)|$; reasoning as above we obtain $|f(e)-g(e)|\leq 2\varepsilon$
    \item $e=(x_i,y_j)$: $|f(e)-g(e)|=|d(x_i,\varphi(y_j))-d(\gamma(x_i),x_j)|$.
    Again in the same fashion we have: $d(x_i,\varphi(y_j))\leq d(x_i,y_j) + d(y_j,\varphi(y_j)) \leq d(x_i,\gamma(x_i)) + d(\gamma(x_i),x_j) + d(y_j,\varphi(y_j)) \leq d(\gamma(x_i),x_j) + 2\varepsilon$. Which entails $|f(e)-g(e)|\leq 2\varepsilon$
\end{itemize}
\end{proof}

\begin{cor}\label{cor:hausdorff}
Given $C=\{x_1,\ldots, x_n\}$ and $C'=\{y_1,\ldots, y_m\}$  point clouds in $(X,d)$ metric space, and given $T_C$ and $T_{C'}$ the single linkage hierarchical clustering dendrograms obtained from $C$ and $C'$ respectively,
we have $d_E(T_C,T_{C'})\leq 4(n+m)d_H(C,C')$.
\begin{proof}
We apply Proposition \ref{prop:cont_functions} and then we are in the position to use Theorem 1 in \cite{pegoraro2021functional} to obtain that $d_E(T_f,T_g)\leq 2(2d_H(C,C'))\cdot (n + m)$.
\end{proof}
\end{cor}

\noindent With the above results, we can prove a last corollary involving the Gromov-Hausdorff distance between compact metric spaces.

\begin{cor}
Given two finite metric spaces $C=\{x_1,\ldots, x_n\}$ and $C'=\{y_1,\ldots, y_m\}$  and given $T_C$ and $T_{C'}$ the single linkage hierarchical clustering dendrograms obtained from $C$ and $C'$ respectively,
we have $d_E(T_C,T_{C'})\leq 4(n+m)d_{G-H}(C,C')$.
\begin{proof}
We apply Proposition \ref{prop:cont_functions} and Corollary \ref{cor:hausdorff} on the images $\gamma(X)$ and $\varphi(Y)$ for every $\gamma:X\rightarrow Z$,
$\varphi:Y\rightarrow Z$ isometries, and for every $Z$ metric space.
\end{proof}
\end{cor}

\section{Proof about $d^{\mu}_P$ being a metric}
\label{sec:proof-metric}

We prove the following proposition.

\begin{prop}
If there is $M>0$ such that for every $m\leq M$, $\mu([0,m])>0$ the $d_\mu^P$ is a metric. 
\end{prop}

\begin{proof}
\end{proof}

\smallskip

\noindent Suppose $d_P^\mu(T,T')=0$. Let $m=\text{min}\{\text{min}_{e\in E_T}w_T(e),\text{ min}_{e'\in E_{T'}}w_{T'}(e')\}$; then for any $\varepsilon\in [0,m)$, $P_\varepsilon(T)=T$ and $P_\varepsilon(T')=T'$. If $d_E(T,T')>0$, since $\mu([0,m))>0$, then:
\[
0<\int_{[0,m)} d_E(P_\varepsilon(T),P_\varepsilon(T'))d\mu(\varepsilon)\leq d_P^\mu(T,T')=0
\]
which is absurd. But then $d_E(T,T')=0$ and so $T=T'$.

\smallskip

\noindent Symmetry is obvious.

\smallskip

\noindent  The triangle inequality holds for $d_E$ and so
\[d_E(P_\varepsilon(T),P_\varepsilon(T')) \leq d_E(P_\varepsilon(T),P_\varepsilon(T'')) + d_E(P_\varepsilon(T''),P_\varepsilon(T'))
\]
The linearity of the integral then entails $d_P^\mu(T,T')\leq d_P^\mu(T,T'')+d_P^\mu(''T,T')$.

\section{Heterogeneity-based Simulation for $d^P_\mu$}
\label{sec:d_mu_simulation}

In this section, we test the metric $d^P_\mu$ and the whole pipeline employed in the case study of the main manuscript in a supervised - in a broad sense - and easier setting.
In particular, the aim of this simulation is to showcase the differences between $d_E$ and $d_\mu^P$ and to which extent $d_\mu^P$ captures heterogeneity in a point cloud.

\smallskip

\noindent We generate point clouds in $\mathbb{R}^2$ according to two generating processes. The size $n^i_1$ of the $i$-th point cloud of the first group is sampled uniformly from $[2,20]\bigcap \mathbb{Z}$ and then a sample of size $(n_1^i, 2)$ is taken from a normal distribution $\mathcal{N}(0, \sigma_1)$, with $\sigma_1 = 1$. Similarly, the $j$-th point cloud of the second group has cardinality $n_2^j$ sampled uniformly from $[2,10]\bigcap \mathbb{Z}$, and the cloud itself is taken as a sample of size 
 $(n_2^j, 2)$ distributed according to $\mathcal{N}(0, \sigma_2)$, with $\sigma_2 = 2$.
The data set of point clouds contains $50$ clouds of the first group and $50$ of the second group.

\smallskip

\noindent From the data-generating processes it is clear that the sources of variability between the two groups arise potentially from the different cardinalities of the point clouds and variance within each cloud. 
We want to show that, while the metric $d_E$ is susceptible to both kind of variability,  $d_\mu^P$, with an appropriately chosen measure $\mu$, can mitigate the variability coming from higher cardinalities in the clouds sampled according to the first process. 
In particular, group 1 is expected to display a lower level of heterogeneity within each point cloud and thus those trees, for our purposes, should be regarded as more similar between each other compared to the other trees. The second group instead may not display a clear clustering structure, in fact, despite exhibiting a common level of heterogeneity, the different number of leaves and the different merging structure at the level of very heterogeneous leaves could prevent all such dendrograms to form a recognizable cluster - or, equivalently, could give birth to a cluster with higher dispersion.

\begin{figure}[h!]
	\centering
	\begin{subfigure}[c]{0.47\textwidth}
    	\centering
    	\includegraphics[width = \textwidth]{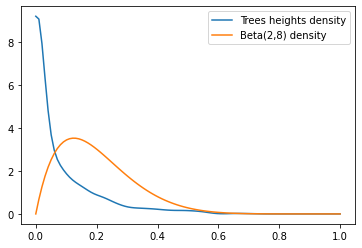}
    	\caption{\textbf{Density of vertices heights from trees in the simulation data, along with the chosen Beta distribution, which has parameters $a=2$ and $b=8$.}}
    	\label{fig:sim_mu}
    \end{subfigure}
	\begin{subfigure}[c]{0.47\textwidth}
		\centering
		\includegraphics[width = \textwidth]{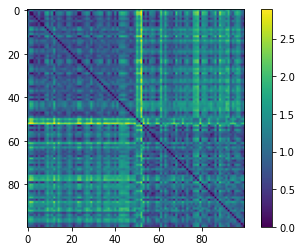}
		\caption{\textbf{Matrix of pairwise distances obtained with $d_E$.}}
		\label{fig:TED_matrix}
	\end{subfigure}
	
		\centering
	\begin{subfigure}[c]{0.47\textwidth}
    	\centering
    	\includegraphics[width = \textwidth]{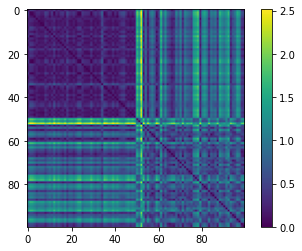}
    	\caption{\textbf{Matrix of pairwise distances obtained with $d^P_\mu$.}}
    	\label{fig:pruned_matrix}
    \end{subfigure}
	\begin{subfigure}[c]{0.47\textwidth}
		\centering
		\includegraphics[width = \textwidth]{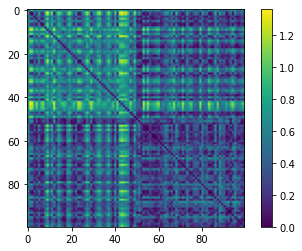}
		\caption{\textbf{Absolute differences between the matrix obtained with $d_E$ and $d^P_\mu$.}}
		\label{fig:diff_matrix}
	\end{subfigure}
	
\caption{The plot in the left upper corner is used to fix $\mu$ in the case study of Section \ref{sec:d_mu_simulation}, according to the procedure detailed in Figure 7 of the manuscript; the other figures show the pairwise distance matrices obtained in the case study of Section \ref{sec:d_mu_simulation}.}
\label{fig:matrices}
\end{figure}

\begin{figure}[h!]
	\centering
	\begin{subfigure}[c]{0.94\textwidth}
    	\centering
    	\includegraphics[width = \textwidth]{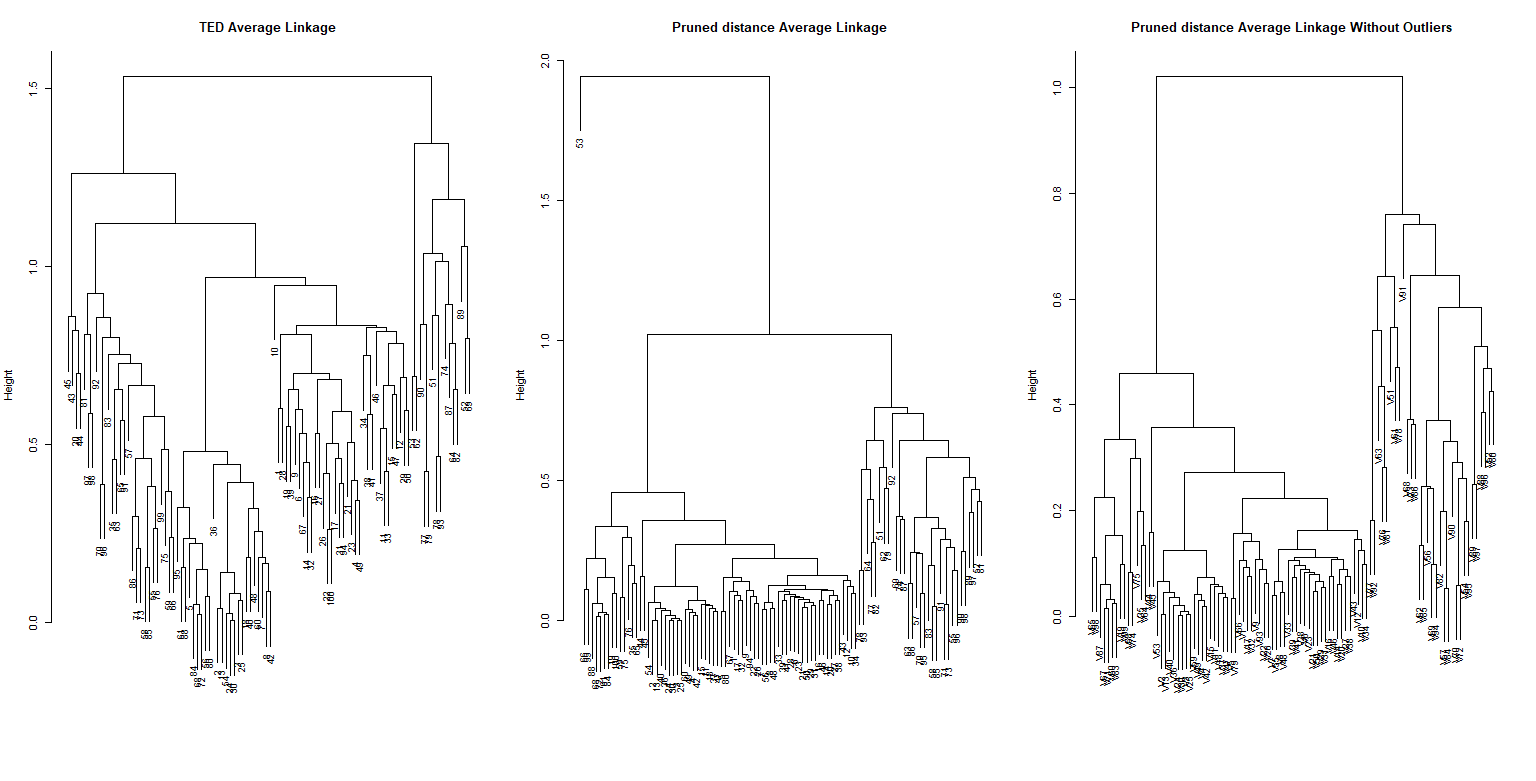}
    	\caption{\textbf{Hierarchical Clustering with average linkage of the pairwise distance matrices respectively obtained from $d_E$, $d_\mu^P$ and $d_\mu^P$ but without the outlier represented by vertex $53$ in the central dendrogram.}}
    	\label{fig:clustering}
    \end{subfigure}

		\centering
	\begin{subfigure}[c]{0.47\textwidth}
		\centering
		\includegraphics[width = \textwidth]{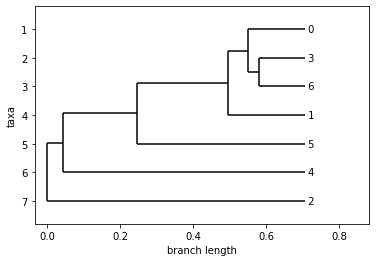}
		\caption{\textbf{The outlier identified by the hierarchical clustering with average linkage of the matrix induced by $d_\mu^P$.}}
		\label{fig:outlier}
	\end{subfigure}
	\begin{subfigure}[c]{0.47\textwidth}
		\centering
		\includegraphics[width = \textwidth]{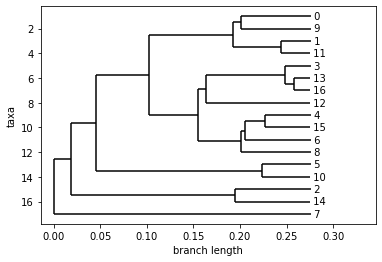}
		\caption{\textbf{A randomly chosen dendrogram belonging to group $2$. The difference in terms of heteogeneity between leaves and number of leaves, with the dendrogram in Fig. \ref{fig:outlier} is evident.}}
		\label{fig:group_1}
	\end{subfigure}	
\caption{Cluster analysis of pairwise distance matrices obtained in the case study of Section \ref{sec:d_mu_simulation}.}
\label{fig:dens_est}
\end{figure}

\smallskip

\noindent Following the pipeline presented in the main manuscript, we extract average linkage hierarchical clustering dendrograms from the set of point clouds and take pairwise distances both with $d_E$ and  
$d_\mu^P$. 
Examples of dendrograms belonging to the first and second groups can be found, respectively, in Fig. S\ref{fig:outlier} and S\ref{fig:group_1}.
We select $\mu$ as in the main manuscript, Section 4.3.2, with the final choice being a Beta distribution with parameters $a=2$, $b=8$, as shown in Fig. S\ref{fig:sim_mu}.
The two matrices are reported in Fig. S\ref{fig:matrices}, with data being ordered according to the two groups: the first $50$ point clouds belong to the first group, and the following $50$ to the second. By visual inspection of Fig. S\ref{fig:TED_matrix} and S\ref{fig:pruned_matrix} we can clearly see that $d_E$ sees very little structure in the data, because of the two sources of variability (cardinality and variance) mixing up and preventing $d_E$ to discriminate between group $1$ and $2$.
Instead $d_\mu^P$ recognizes a clear and pronounced cluster made by point clouds from group $1$ plus, potentially, some other point clouds belonging to group $2$.
The rest of the point clouds of group $2$ still show some agglomerative structure, but less evident. The matrix in Fig. S\ref{fig:diff_matrix} shows the pointwise differences between the values obtained with $d_E$ and $d_\mu^P$, highlighting how the different behaviour of the two metrics concentrates on the data belonging to the first group.

\begin{figure}[h!]
	\centering
	\begin{subfigure}[c]{0.47\textwidth}
    	\centering
    	\includegraphics[width = \textwidth]{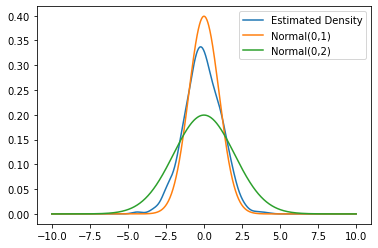}
    	\caption{E\textbf{stimated density of the first component of the data in the first cluster identified by $d_\mu^P$, versus the densities generating the samples two groups.}}
    	\label{fig:dens_1_x}
    \end{subfigure}
	\begin{subfigure}[c]{0.47\textwidth}
		\centering
		\includegraphics[width = \textwidth]{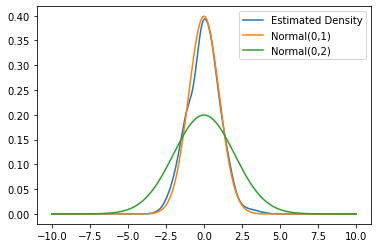}
		\caption{\textbf{Estimated density of the second component of the data in the first cluster identified by $d_\mu^P$, versus the densities generating the samples two groups.}}
		\label{fig:dens_1_y}
	\end{subfigure}
	
		\centering
	\begin{subfigure}[c]{0.47\textwidth}
    	\centering
    	\includegraphics[width = \textwidth]{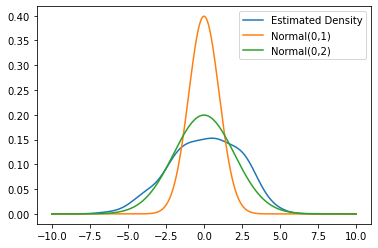}
    	\caption{\textbf{Estimated density of the first component of the data in the second cluster identified by $d_\mu^P$, versus the densities generating the samples two groups.}}
    	\label{fig:dens_2_x}
    \end{subfigure}
	\begin{subfigure}[c]{0.47\textwidth}
		\centering
		\includegraphics[width = \textwidth]{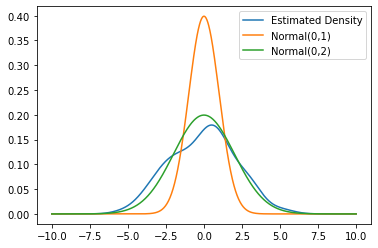}
		\caption{\textbf{Estimated density of the second component of the data in the second cluster identified by $d_\mu^P$, versus the densities generating the samples two groups.}}
		\label{fig:dens_2_y}
	\end{subfigure}
	
\caption{Densities estimated through the aggregation of the data collected in the two clusters identified by $d_\mu^P$.}
\label{fig:est_dens}
\end{figure}

\smallskip

\noindent To get more insights into the clustering structures expressed by $d_E$ and $d_\mu^P$ we extract the hierarchical clustering dendrograms with average linkage from the two matrices. These dendrograms are reported in Fig. S\ref{fig:clustering}. 
The leftmost tree is obtained from $d_E$ and the central from $d_\mu^P$. To better compare the clustering structures we remove from this last dendrogram the outlier ($v53$), obtaining the rightmost tree. 

\smallskip

\noindent Visual inspection of the dendrograms in Fig. S\ref{fig:clustering} reveals a two-clusters structure in both metric spaces, with this structure being much more recognizable in the metric space induced by $d_\mu^P$. In particular, the rightmost dendrogram shows a very cohesive and compact cluster, with very low internal variability, which is absent in the leftmost tree. The other cluster of the same tree, instead, displays a much higher level of variability.

\smallskip

\noindent Now we show that this clustering structure reflects the group structure that generated our data.
We cut the rightmost tree to obtain two clusters. Then, for each cluster, we aggregate the points contained in the data of such cluster and we estimate the marginal densities from the obtained samples. 
The results of this estimation pipeline is showcased in Fig. S\ref{fig:est_dens}.
We see that we retrieve the two distributions which we used to generate the components of the point clouds of the two groups.

\smallskip

\noindent This is precisely the behaviour we aimed to achieve: being insensitive to the cardinality of small homogeneous features, while still being sensitive to cardinalities and merging structures characterized by high heterogeneity.

\section{Additional plots for clustering interpretation}
\label{sec:additional}

Fig. \ref{fig:additional} integrates the results, in terms of cluster characterization.

\begin{figure}[t]
	\centering
    	\includegraphics[width = \textwidth]{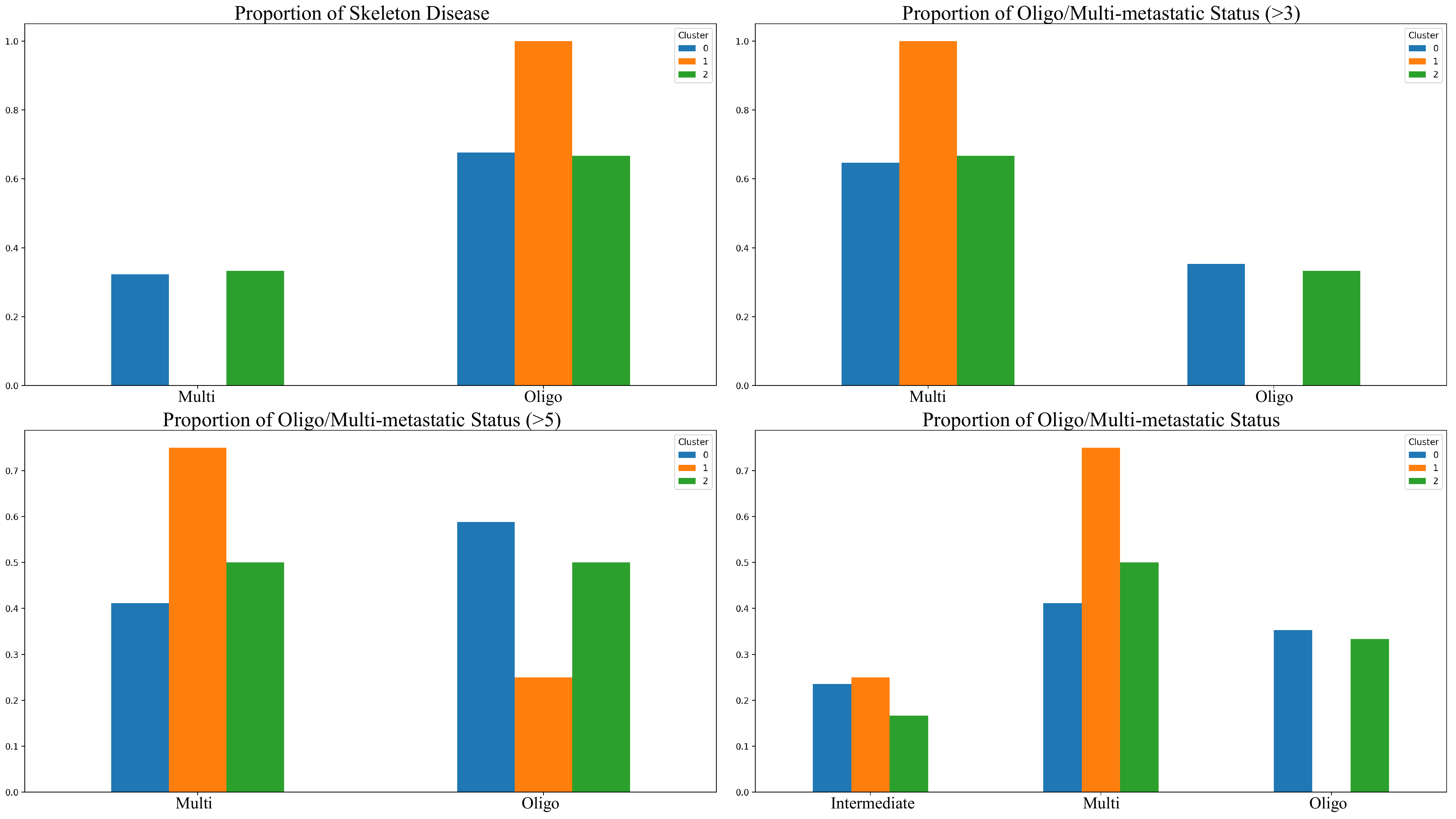}
\caption{Results of clustering characterization: the proportion of skeleton disease and of the oligo/multi-metastatic status as devised by the two clinical cut-offs (3 and 5 lesions) are plotter per each of the three group.}
\label{fig:additional}
\end{figure}

\newpage

\section*{Data and code availability}
The data supporting the findings of this study are available from Azienda Ospedaliero-Universitaria Pisana but restrictions apply to the availability of these data, which were used under license for the current study, and so are not publicly available. 

\noindent The code implemented during the current study together with simulation data are available on github at this \href{https://github.com/pego91/pruned-edit-distance}{link}.

\section*{Acknowledgements (not compulsory)}
This work was carried out as part of the PhD Thesis of Lara Cavinato, under the supervision of Professor Francesca Ieva, and part of the PhD Thesis of Matteo Pegoraro, under the supervision of Professor Piercesare Secchi from Politecnico di Milano.
We acknowledge all the personnel of Medicine Department of Azienda Ospedaliero-Universitaria Pisana for the assistance during the PET/CT scans, segmentation of lesions, extraction of radiomic features and retrieval of patients’ personal information from EHR.

\section*{Author contributions statement}
L.C. conceived the pipeline, set up the case study, analyzed the results, prepared the figures, and wrote the manuscript.
M.P. formulated and tuned the pruned tree edit distance, provided the mathematical proofs and the simulation study, and wrote the manuscript.
A.R. contributed to implement the patient representation pipeline.
M.S. segmented the Prostate Cancer lesions and extracted the radiomic features for all patients in the case study.
P.A.E. collected the data and enrolled the patients in the clinical study.
F.I. supervised the analyses and the conception of the pipeline.
L.C., M.P., A.R. and F.I. reviewed and approved the manuscript. 

\section*{Additional information}

\textbf{Competing interests} The author(s) declare no competing interests.

\smallskip

\noindent \textbf{Editorial process} This version of the article has been accepted for publication, after peer review (when applicable) but is not the Version of Record and does not reflect post-acceptance improvements, or any corrections. The Version of Record is available online at this \href{https://doi.org/10.1038/s41598-022-23752-2}{DOI}.

\newpage

\bibliography{manuscript}

\end{document}